\documentclass[11pt]{amsart}
\usepackage{mypackage}
\usepackage{myformat}

\usepackage{hyperref}
\hypersetup{colorlinks=true,linkcolor=cyan}

\newcommand {\Neg}{\mathscr N}
\newcommand {\W}{W_{_\succ}}
\newcommand {\K}{\mathcal K}

\newcommand {\Ca}{\mathscr C}
\newcommand {\CapiA}{\mathscr C(\pi,\A)}

\newcommand {\Xa}{\mathscr X}
\newcommand {\Meas}{\mathscr M}    
\DeclareMathOperator{\Ext}{Ext}
\DeclareMathOperator{\Int}{Int}
\newcommand{\ma}{\mathfrak m}
\newcommand{\na}{\mathfrak m^c}

\begin{document} 
\title[Market Completeness]
{Complete and competitive financial markets in 
a complex world}

\mybic
\date \today 


\begin{abstract} 
We investigate the possibility of completing financial
markets in a model with no exogenous probability
measure, with market imperfections and with an
arbitrary sample space. We also consider whether 
such extension may be possible in a competitive 
environment. Our conclusions highlight the economic
role of complexity.

\end{abstract}

\thanks{
I am grateful to the participants at SAET 2019 Conference, 
Ischia, for comments and suggestions.
}

\maketitle

\section{Introduction.}
Since the seminal contributions of Arrow \cite{arrow} 
and of Radner \cite{radner}, market completeness and 
the no arbitrage principle have played a prominent r\^ole 
in financial economics. Market completeness, as first 
noted by Arrow, is a crucial property as it permits the 
optimal allocation of risk bearing among risk averse 
agents. In fact the equilibria of an economy under 
conditions of uncertainty but with competitive and 
complete financial markets are equivalent to those of 
an ordinary static economy so that classical welfare 
theorems apply. The equilibrium analysis on which this 
conclusion rests requires that financial markets are 
free of arbitrage opportunities. In the present paper 
we consider the validity of these classical results in 
a complex world, that is in the context of an economic 
model in which uncertainty is treated as a completely 
general and unrestricted phenomenon. As a general
conclusion we find that, in a complex world the 
interplay between uncertainty and asset prices is
richer and more interesting than expected.

Indeed the model of Arrow and of Radner, and of much 
of the following literature on general equilibrium theory 
with financial markets, focuses on an economy with a 
finite state space. This modelling choice is instrumental 
to and has its main advantage in the description of 
financial assets as contingent contracts, i.e. in purely
economic terms. On the other hand, it has the drawback
that this simplified representation of uncertainty makes 
it more challenging, in the absence of further elements, 
to justify market incompleteness. To the other extreme, 
in infinite dimensional economic theory (see e.g. the 
review of Mas-Colell and Zame \cite{mas-colell_zame}) 
commodities and assets need to be defined as elements
of some given space, the choice of which is most often 
a first step of crucial importance.

In the approach we propose hereafter we retain Arrow's
original idea that assets should be described solely in 
terms of the rights and obligations of the two counterparties
on the occurrence of each future contingency and yet 
allow for an arbitrary sample space $\Omega$. In particular, 
$\Omega$ will not be endowed with any special structure 
and the real valued functions $X$ defined thereon, which 
describe assets payoffs, need not be continuous nor 
measurable in any specific sense. In fact, following the 
thread of \cite{MAFI_2008} and \cite{ECTH_2017}, we 
do not assume the existence of any exogenously given 
probability measure, a major difference with a large 
part of the financial literature as well as many important 
papers in equilibrium theory, such as Bewley \cite{bewley}. 
Our starting point is rather a criterion of economic rationality 
which describes what all agents agree upon when saying 
that ``$f$ is more than $g$'' (this modelling of rationality, 
first introduced in \cite{ECTH_2017}, is referred to as 
{\it common order} in \cite{burzoni_riedel_soner}). 

We believe that our general framework is indeed the
natural setting for assuming incomplete markets and
to investigate the main aspects of a process of gradual 
completion of markets. In particular, we address the 
following questions:
\tiref a
can an incomplete set of financial markets be extended 
to a complete one while preserving the basic economic
principle of absence of arbitrage opportunities?
\tiref b
if so, can such an extension be supported by a
competitive market mechanism?

Our finding is that neither question need have a
positive answer. Concerning the former, we observe 
that competition on financial markets may in principle 
produce two distinct outcomes. On the one hand 
it lowers margins on currently traded assets and 
results thus in lower prices. On the other hand, 
competition involves the design and issuance of 
new securities. We argue that lower prices on the 
existing securities may destroy the possibility to 
obtain complete markets free of arbitrage opportunities 
in much the same way as predatory pricing in a 
monopolistic setting may prevent the entry of new, 
potential competitors. In principle the net effect of 
competition on collective welfare may be unclear. 
Relatively to the latter question, we show that the 
completion of financial markets in respect of the no 
arbitrage principle may not be possible under linear 
pricing (which we take as synonymous of perfect 
competition). We actually provide an explicit example 
from which it clearly emerges that this negative 
conclusion has to do with the complexity of the 
economic environment described in the model. 
Indeed, most economic models treat economic 
complexity under probabilistic assumptions which 
do not permit a clear comprehension of this phenomenon. 
On the other hand, we show that if a no arbitrage 
extension of markets with a limited degree of market 
power is possible, then markets admit a perfectly 
competitive extension as well.

We should make clear that, although it is indeed
natural and appropriate on a general ground, to
interpret the extension of markets as the effect
of financial innovation, an explicit model of the 
strategic behaviour of intermediaries, such as Allen 
and Gale \cite{allen_gale} or Bisin \cite{bisin}, is 
beyond the scope of the present work. We rather 
study the properties of pricing functions described 
as a sublinear functional on the space of traded 
assets' payoffs. The non linearity of prices captures 
the non competitive nature of financial markets as 
well as the role of other market imperfections.


In recent years there have been several papers in which 
the assumption of a given reference probability is relaxed, 
if not abandoned. Riedel \cite{riedel} (and more recently 
Burzoni, Riedel and Soner \cite{burzoni_riedel_soner}) 
suggests that an alternative approach to finance should 
be based on the concept of Knightian uncertainty. A typical 
implication of this approach is that a multiplicity of 
probability priors is given -- rather than a single one. 
Some authors, including Bouchard and Nutz 
\cite{bouchard_nutz}, interpret this multiplicity as an 
indication of {\it model uncertainty}, a situation in which 
each prior probability corresponds to a different model 
that possesses all the traditional properties but in which
it is unknown which of the models should be considered
the correct one. An exemplification is the paper by Epstein 
and Ji \cite{epstein_ji} in which model uncertainty 
translates into ambiguity concerning the volatility
parameter. Other papers, among which the ones by 
Davis and Hobson \cite{davis_hobson} and by Acciaio 
et al \cite{acciaio_et_al}, take the sample space to 
consist of all of the trajectories of some underlying asset
and study the prices of options written thereon based on 
a path by path or {\it model-free} definition of arbitrage.

In section \ref{sec economy} we describe the model in
all details, we introduce the notion of arbitrage and 
prove some properties of prices. In section \ref{sec 
pricing measures} we characterize the set of pricing 
probability charges which is of crucial importance in 
our construction. In the following section \ref{sec 
complete} we prove the first of our main results, Theorem 
\ref{th NFLVR}, in which the possibility of completing 
markets is fully characterized. We show that in a complete 
financial market, although prices may in principle contain 
bubbles, there cannot be assets with a positive price but 
no intrinsic value. Theorem \ref{th NFLVR} provides an 
answer to question \tiref a above. In section \ref{sec 
viable} we examine viability of financial prices which,
in our setting, turns out to be a stronger property than 
the extension property. In section \ref{sec competition} 
we establish a second fundamental result, Theorem \ref{th NFM}, 
which answers question \tiref b. It offers a complete 
characterization of the existence of a fully competitive 
completion of financial markets. In particular we prove 
that a competitive completion is possible if only one 
can obtain a completion in which the market power of 
intermediaries is limited. This result can also be read 
as a theoretical justification of the microstructure formula 
stating that asset prices are obtained by applying a 
spread on the asset fundamental value. Several additional 
implications are proved. We provide an explicit example 
of a financial market that admits no competitive completions 
as a consequence of a high degree of complexity. Given 
its importance in the reference literature, in section 
\ref{sec c.a.} we examine the question of countable 
additivity and give exact necessary and sufficient
conditions for the existence of a countably additive
pricing probability.

\section{The Economy.}
\label{sec economy}

We model the market as a triple, $(\Xa,\ge_*,\pi)$, in 
which $\Xa$ describes the set of payoffs generated
by the traded assets, $\ge_*$ the criterion of collective 
rationality used in the evaluation of investment projects
and $\pi$ is the price of each asset as a function of its 
payoff. Each of these elements will now be described 
in detail.

Before getting to the model we introduce some useful
notation.
Throughout, $\Omega$ will be an arbitrary, non empty 
set that we interpret as the sample space so that the 
family $\Fun\Omega$ of real valued functions on 
$\Omega$ will be our ambient space; $\B(\Omega)$ 
will represent bounded functions. If $A\subset\Fun\Omega$, 
we write $\cl[u]A$ (resp. $\cl[\tau]A$) to denote its 
closure in the topology of uniform distance (resp. in 
the topology $\tau$). If $\A$ is a $\sigma$ algebra 
of subsets of $\Omega$, $\Fun{\Omega,\A}$ denotes 
the family of $\A$ measurable functions and we set
$\B(\Omega,\A)=\Fun{\Omega,\A}\cap\B(\Omega)$.
The
symbol $ba(\A)$ denotes the set of bounded, finitely 
additive set functions on $\A$, to which we refer as 
{\it charges}, while $\Prob_{ba}(\A)$ designates the 
collection of probability charges. We reserve the word 
probability with no further qualification and the symbol 
$\Prob(\A)$ to classical (i.e. countably additive) probabilities. 
General references for the theory of charges and 
finitely additive integrals are \cite{bible} and \cite{rao}.

\subsection{Economic Rationality}
\label{subsec	rationality}
A natural order to assign to $\Fun\Omega$ is pointwise
order, to wit 
$f(\omega)
	\ge 
g(\omega)$ 
for all $\omega\in\Omega$, also written as $f\ge g$. The 
lattice symbols $\abs f$ or $f^+$ will always refer to such 
order.

Natural as it may appear, pointwise order is not an 
adequate description of how economic agents rank 
random quantities save when the underlying sample 
space is particularly simple, such as a finite set. For 
example, it is well documented that investors base 
their decisions on a rather incomplete assessment of 
the potential losses arising from the selected portfolios, 
exhibiting a sort of asymmetric attention that leads 
them to neglect some scenarios, in contrast with a 
pointwise ranking of investment projects%
\footnote{
See \cite{MAFI_2008} for a short discussion of some 
inattention phenomena relevant for financial decisions.
}.
In a complex world, in which the attempt to formulate
a detailed description of $\Omega$ is out of reach,
rational inattention is just one possible approach to
deal with complexity, possibly not too different from
the restriction to measurable quantities adopted by
classical probability. 

In this paper, following the thread of \cite{ECTH_2017}, 
we treat monotonicity as a primitive economic notion 
represented by a further transitive, reflexive binary 
relation on $\Fun\Omega$. To distinguish it from
the pointwise order $\ge$, we use the symbol $\ge_*$
(the asymmetric part of $\ge_*$ will be written as
$>_*$).

We assume the following properties:
\begin{subequations}
\label{ge}
\begin{equation}
\label{ge mon}
(i).\quad
f\ge g
\qtext{implies}
f\ge_* g
\qtext{and}
(ii).\quad
t\in\R_{++}
\qtext{implies}
f+t>_*f,
\end{equation}
\begin{equation}
\label{ge affine}
\qtext{if}
f\ge_* g
\qtext{then}
\lambda f+h\ge_*\lambda g+h,
\qquad
\lambda\in\B(\Omega)_+, h\in\Fun\Omega
\end{equation}
\begin{equation}
\label{ge trunc}
\qtext{if}
f>_*0
\qtext{then}
f\wedge 1>_*0
\end{equation}
\end{subequations}
which will be the basis for what follows%
\footnote{
The first paper to treat monotonicity in an axiomatic way 
was, of course, Kreps \cite{kreps}. In \cite{ECTH_2017} 
axioms similar to \eqref{ge} were first introduced and
it was show that a ranking with such properties may 
equivalently be deduced from a cash additive risk 
measure. A much similar approach to ours has later
been adopted in the recent paper by Burzoni et al 
\cite{burzoni_riedel_soner}. 
}.

Of course, the symmetric part of $\ge_*$ induces 
a corresponding equivalence relation, $\sim_*$. It 
will be useful to remark that if $f\ge_*0$ then, by 
property \eqref{ge affine}, $f\sset{f\le0}\ge_*0$ so 
that 
$f
	\ge_*
f-2f\sset{f\le0}
	=
\abs f$,
i.e. $\abs f\sim_*f$. Based on $\sim_*$ we can 
also construct the collection of {\it negligible sets}
\begin{equation}
\label{Neg}
\Neg_*
	=
\{A\subset\Omega:0\sim_*\set A\}
\end{equation}
and, if $\A$ is an arbitrary $\sigma$ algebra of subsets 
of $\Omega$ which contains $\Neg_*$, the subset
$\Prob_{ba}(\A,\Neg_*)$ consisting of probability 
charges on $\A$ which vanish on $\Neg_*$. Every 
subset of $\Omega$ not included in $\Neg_*$ will 
be called {\it non negligible}. It will be useful to
define the space
\begin{equation}
\B(\Omega,\A,\Neg_*)
	=
\big\{f\in\Fun{\Omega}:
f\sim_*b\text{ for some }b\in\B(\Omega,\A)\big\}.
\end{equation}

It is immediate to note that any {\it exogenously} 
given probability charge $P\in\Prob_{ba}(\A)$ on
some $\sigma$ algebra $\A$ (countably or finitely 
additive) induces a corresponding ranking defined 
as
\begin{equation}
\label{ge P}
f\ge_{_P} g
\qtext{if and only if}
\sup_{\varepsilon>0}P(f\le g-\varepsilon)=0
\qquad
f,g\in\Fun{\Omega,\A},
\end{equation}
which satisfies the above axioms \eqref{ge}. In 
this case a set is negligible if and only if it is
$P$ null. Clearly,
if $P$ is a probability the ranking $\ge_{_P}$ is just
the $P$-a.s. ranking. The same construction may be
extended by replacing $P$ with a family 
$\Pred\subset\Prob_{ba}(\A)$ of 
probability charges and defining accordingly
\begin{equation}
\label{ge PP}
f\ge_{_\Pred} h
\qtext{if and only if}
f\ge_{_P} h
\qtext{for all}
P\in\Pred.
\end{equation}
The ranking $\ge_{_\Pred}$ defined in \eqref{ge PP} 
arises in connection with the {\it model uncertainty} 
approach mentioned in the Introduction and 
exemplified by the paper by Bouchard and Nutz
\cite{bouchard_nutz}. In this approach each element 
$P$ of the given collection $\Pred$ is a {\it model}%
\footnote{
It should be noted that the choice of Bouchard and 
Nutz to take $\Pred$ to be a set of probabilities has 
considerable implications on $\Neg_*$ which coincides 
with the collection of sets which are $P$ null for every 
$P\in\Pred$ and is therefore closed with respect to 
countable unions. We shall return on the ranking
$\ge_\Pred$ in Example \ref{ex PP}.
}.

We provide some concrete examples of the partial
order $\ge_*$ that arise from decision theory.

\begin{example}
\label{ex collective}
Let $\succeq_\alpha$ represent the preference system 
of agent $\alpha$ defined over the whole of $\Fun\Omega$.
Assume that $\succeq_\alpha$ satisfies the following
monotonicity properties valid for all $f,g\in\Fun\Omega$%
\footnote{
These assumptions are indeed minimal in economic
models with an infinite dimensional commodity space,
see e.g. Bewley \cite[p. 520]{bewley} or Mas-Colell 
\cite[p. 1041]{mas-colell}. 
}:
\begin{equation}
\label{pref a}
(i).\quad
f\ge g
\qtext{implies}
f\succeq_\alpha g
\qtext{and}
(ii).\quad 
t\in\R_{++}
\qtext{implies}
f+t\succ_\alpha f.
\end{equation}
We may deduce an implicit, 
subjective criterion $\ge_\alpha$ of monotonicity (or 
rationality), by letting
\begin{equation}
\label{end}
f\ge_\alpha g
\qtext{if and only if}
\lambda(f-g)+h
\succeq_\alpha h
\qquad
\lambda,h\in\Fun\Omega,\ 
\lambda\ge0.
\end{equation}
A mathematical criterion $\ge_*$ describing collective 
rationality may then be defined as the meet of all such 
individual rankings, i.e. as
\begin{equation}
\label{unanimity}
f\ge_*g
\qtext{if and only if}
f\ge_\alpha g
\qtext{for each agent}
\alpha\in\mathfrak A.
\end{equation}
Notice that in this case $f>_*g$ means that all agents
agree on $f\ge_\alpha g$ and there exists at least
one agent $\alpha_0\in\mathfrak A$ according to 
which $f>_{\alpha_0}g$.
\end{example}

\begin{example}
\label{ex EU}
Let an agent $\alpha$ have expected utility preferences
$U_\alpha(f)
=
\int u_\alpha(f)dP_\alpha$
with $u_\alpha$ a monotonic increasing utility function 
satisfying $u_\alpha(t)>u_\alpha(0)=0$ for each $t>0$. 
Define $\ge_\alpha$ as in the preceding Example 
\ref{ex collective}. Then $f\ge_\alpha g$ implies, by 
\eqref{end}, that
$U_\alpha(0)
	\ge 
U_\alpha((f-g+t)\sset{f-g<-t})
	\ge 
U_\alpha(t\sset{f-g<-t})
	=
u_\alpha(t)P_\alpha(f-g<-t)$
i.e. that $f\ge_{P_{\alpha}}g$. The converse is, of 
course, also true. Thus, $\ge_\alpha$ coincides with
$\ge_{P_\alpha}$. If we aggregate the $\ge_\alpha$
rankings via the unanimity rule \eqref{unanimity},
we obtain thus the ranking $\ge_{\Pred}$ defined
in \eqref{ge PP}, with $\Pred=\{P_\alpha:\alpha\in\mathfrak A\}$.
\end{example}

W shall return on Example \ref{ex EU} in Section 
\ref{sec competition}, Example \ref{ex PP}.

\begin{example}
\label{ex choquet}
In the theory of decision making, several models,
starting with \cite{schmeidler_89}, have treated
the problem of choice under uncertainty by means
of set functions with weak properties, often not 
even additive. Fix to this end 
$\nu:2^\Omega\to[0,1]$ and
define $\ge_\nu$ as in \eqref{ge P}. Then $\ge_\nu$
satisfies properties \eqref{ge} if and only if
(i)
$\nu(\emp)=0$
and
(ii)
$\nu(A)\le\nu(B)$ when $A\subset B$, i.e. when 
$\nu$ is a capacity. Nevertheless, the binary
relation $\ge_\nu$ (which is always reflexive
by virtue of (i)) is transitive if and only if $\nu$ 
satisfies, in addition, the further property
\begin{equation}
\label{null additive}
\nu(A\cup B)=\nu(A)
\qquad
B\subset\Omega,\ \nu(B)=0,
\end{equation}
which is sometimes referred to as null-additivity. A 
strongly subadditive (or submodular) capacity
\cite[III.30]{dellacherie_meyer_A} clearly 
possesses all the above properties while a 
supermodular one need not satisfy 
\eqref{null additive}.
\end{example}

The conclusion of the preceding Example \ref{ex choquet}
can be given a formal statement in the following:

\begin{lemma}
\label{lemma choquet}
A ranking that satisfies property \eqref{ge} originates 
from a set function $\nu:2^\Omega\to[0,1]$ via 
\eqref{ge P} if and only if $\nu$ is a null-additive
capacity.
\end{lemma}

\begin{example}
\label{ex filter}
In behavioural economics and experimental 
psychology it is often argued that decisions
are taken by focusing attention on a restricted
class of possible scenarios to which agents
attach special importance or simply over
which they feel relatively confident or just
more informed. Let $\F$ be a family of non 
empty subsets of $\Omega$, representing
the sets relatively to which the agent is
fully confident. An agent may then rank 
alternatives in accordance with the
following criterion:
\begin{equation}
\label{ge F}
f\ge_\F h
\qtext{if and only if }
\{f-h>-t\}\in\F
\qquad
t>0.
\end{equation}
Again, one easily concludes that $\ge_\F$
possesses the above properties \eqref{ge}
if and only if $\F$ is a filter of subsets of
$\Omega$. In accordance with \cite[I.7.11]{bible} 
let $\F_*$ be some ultrafilter refining $\F$, 
and define $P\in\Prob_{ba}(2^\Omega)$ by letting 
$P(G)=1$ if $G\in\F_*$ or else $P(G)=0$.
Then $f\ge_\F h$ implies $f\ge_{_P}h$;
the converse implication need not be true.
\end{example}

An interesting question raised by the preceding 
examples concerns the conditions under which 
a given ranking $\ge_*$ coincides with the
ranking $\ge_{_P}$ for some {\it endogenous} 
probability charge $P$. Notice that in Example
\ref{ex filter} we only showed that $\ge_{_P}$
is compatible with $\ge_\F$ but we may still 
have that $f\ge_{_P}h$ while $f\not\ge_\F h$.
We shall obtain an indirect answer to this
question in Theorem \ref{th L1}.

\subsection{Assets}
\label{subsec assets}
We posit the existence of an asset whose final payoff 
and current price are used as {\it num\'eraire} of the 
payoff and of the price of all other assets, respectively. 
Each asset is identified with its payoff expressed in units 
of the {\it num\'eraire} and is modelled as an element 
of $\Fun\Omega$. The market is then a convex set 
$\Xa
	\subset 
\Fun\Omega$ 
containing the origin as well as the function identically 
equal to $1$ (that will be simply indicated by $1$). 
Notice that, similarly to Jouini and Kallal
\cite{jouini_kallal_99}, we do not assume that investments 
may be replicated on any arbitrary scale, i.e. that $\Xa$ is 
a {\it convex cone}, as is customary in this literature.
Restrictions to the investment scale are commonly met
on the market whenever investors are subject to margin
requirements or the provision of a collateral.

We assume in addition that 
({\it i})
each $X\in\Xa$ satisfies $X\ge_*a$ for some $a\in\R$ 
and 
({\it ii}) that
\begin{equation}
\label{safe}
X+\lambda\in\Xa
\qtext{for all}
X\in\Xa,\ \lambda\ge0.
\end{equation}
The first of these assumptions constraints the assets 
traded on the market to bear a limited risk of losses
and may be interpreted as a restriction imposed by
some regulator; the second one permits agents, which 
in principle may only form convex  portfolios, to invest 
into the {\it num\'eraire} asset an unlimited amount of 
capital. Notice that, since the {\it num\'eraire} cannot 
be shorted, the construction of zero cost portfolios -- 
or {\it self-financing strategies} -- need not be possible
(outside of the trivial case of assets with a negative price
which are self-financed by definition). 
We do not {\it assume} that the market prohibits short
positions but rather that, in the presence of credit
risk, long and short positions even if permitted should
be regarded as two different investments as they bear
potentially different levels of risk. In other words, when 
taking short positions, investors affect the implicit 
counterparty risk and modify {\it de facto} the final 
payoff of the asset shortened.

\subsubsection{Market Extensions}
\label{subsec complete}

The issuance of new securities may result in the extension 
of the set $\Xa$ of traded assets and, possibly, in the
completion of the existing markets. Let us mention that
in traditional models, in which asset payoffs are modelled
as elements of an appropriate ambient space -- often
a Lebesgue space $L^p(	\Omega,\F,P)$ -- markets are
considered to be complete if each elements of such space
is attained by some traded asset. In our setting, defining
market completeness is more delicate. The idea that all 
elements of $\Fun\Omega$ may be traded is indeed too
ambitious as it would be difficult to define a price function
on such a large domain.

In order to have an appropriate notion of completeness
of markets, we find it convenient to introduce a $\sigma$ 
algebra $\A$ of subsets of $\Omega$ satisfying 
\begin{equation}
\label{A}
\bigcup_{\{f\sim_* X:X\in\Xa\}}
\sigma(f)
\subset\A
\end{equation}
and to define
\begin{equation}
\label{L(X,A)}
L(\Xa,\A)
	=
\big\{g\in\Fun{\Omega,\A}:
\lambda X\ge_*\abs g\text{ for some }\lambda>0,\ 
X\in\Xa\big\}
\end{equation}
which may be loosely interpreted as the set of ($\A$ 
measurable) superhedgeable claims. We then define 
markets to be complete relatively to $\A$ (or $\A$
complete, for brevity) whenever each element of
$L(\Xa,\A)$ is attained by a corresponding asset.
In other words, we define market completeness
{\it conditionally} on the $\sigma$ algebra $\A$
and of course a whole hierarchy of completions 
is possible as $\A$ ranges between the two extrema 
of the minimal $\sigma$ algebra $\A_0$ generated 
by the elements equivalent to some $X\in\Xa$ and 
the power set of $\Omega$. The introduction of a 
measurable structure, which implies no restriction in 
our analysis, serves two distinct purposes: on the one 
hand, it makes the comparison with the traditional
literature fully transparent; on the other hand, it 
permits to put the extension problem in a clearer 
relation with the notion of complexity, that we identify
with the fineness of $\A$. A fully complex model
corresponds then to the case in which $\A$ is the 
power set of $\Omega$, a case perfectly possible
in our setting.

\begin{lemma}
The set $L(\Xa,\A)$ defined in \eqref{L(X,A)} is a vector
sublattice of $\Fun{\Omega,\A}$ containing $\Xa$ and
$\B(\Omega,\A,\Neg_*)$.
\end{lemma}

\begin{proof}
The last inclusion is clear from the definition of
$L(\Xa,\A)$. It is easily seen that $g\in L(\Xa,\A)$ 
if and only if $\abs g\in L(\Xa,\A)$;
moreover, if $\lambda_iX_i\ge_*\abs{g_i}$ and if
$a_i\in\R$ for $i=1,2$ then, upon letting
$\lambda
	=
\abs{a_1}\lambda_1+\abs{a_2}\lambda_2$ 
and
$X
	=
X_1(\abs{a_1}\lambda_1/\lambda)
+
X_2(\abs{a_2}\lambda_2/\lambda)$,
we conclude, from \eqref{ge affine},
\begin{align*}
\lambda X
=
\abs{a_1}\lambda_1 X_1+\abs{a_2}\lambda_2X_2
\ge_*
\abs{a_1}\abs{g_1}+\abs{a_2}\abs{g_2}
\ge
\abs{a_1g_1+a_2g_2}
\end{align*}
so that $g_1,g_2\in L(\Xa,\A)$ and $a_1,a_2\in\R$
imply
$a_1g_1+a_2g_2\in L(\Xa,\A)$. This proves that
$L(\Xa,\A)$ is a vector sublattice of $\Fun{\Omega,\A}$.
If $Y\in\Xa$, then $Y+a\ge_*0$ for some $a\in\R_+$
so that $Y+2a\in\Xa$ and
$Y+2a
	\ge_*
\abs{Y+a}+a
	\ge
\abs Y
$
so that $Y\in L(\Xa,\A)$. 
\end{proof}

When markets are $\A$-complete each element 
$f\in\B(\Omega,\A,\Neg_*)$ corresponds to the
payoff of some asset traded on the market. In the 
special case in which $\ge_*$ is generated by 
some $P\in\Prob(\A)$, market completeness 
implies that all claims with payoff in 
$L^\infty(\Omega,\F,P)$ are attained.

\subsection{Prices.}
\label{subsec prices}
In financial markets with frictions and limitations to 
trade, normalized prices are best modelled as positively 
homogeneous, subadditive functionals of the asset payoff, 
$\pi:\Xa\to\R$, satisfying $\pi(1)=1$, the 
monotonicity condition
\begin{equation}
\label{NA mon}
X,Y\in\Xa,\ X\ge_*Y
\qtext{imply}
\pi(X)\ge\pi(Y)
\end{equation}
and cash additivity%
\footnote{
This property, defined in slightly different terms, is 
discussed at length relatively to risk measures in 
\cite{el-karoui_ravanelli}. In the context of non linear
pricing cash additivity is virtually always assumed in
a much stronger version, namely for all $X\in\Xa$ and
all $a\in\R$, see e.g. \cite[Definition 1]{araujo_et_al}.
}
\begin{equation}
\label{cash add}
\pi(X+a)
	=
\pi(X)+a
\qquad
X\in\Xa,\ 
a\in\R
\qtext{such that}
X+a\in\Xa.
\end{equation}
The non linearity of financial prices is a well known 
empirical feature documented in the microstructure 
literature (see e.g. the exhaustive survey by Biais et al. 
\cite{biais_glosten_spatt}) and essentially accounts 
for the auxiliary services that are purchased when 
investing in an asset, such as liquidity provision and 
inventory services. Subadditivity captures the idea 
that these services are imperfectly divisible and that
they are supplied in conditions of limited competition. 
The degree of non linearity of prices, computed as
\begin{equation}
\label{market power}
\ma(\pi)
	=
\sup_{f_1,\ldots,f_N\in\B(\Omega,\A,\Neg_*)_+\cap\Xa}
\frac{\sum_{i\le N}\pi(f_i)-\pi\left(\sum_{i\le N}f_i\right)}
{\sum_{i\le N}\pi(f_i)},
\end{equation}
will play a major role in section \ref{sec competition}
and it measures the lack of competitiveness on financial 
markets. Although deviations from the competitive
paradigm may originate from several sources, we 
find it convenient to interpret the polar cases 
$\ma(\pi)=0$ and $\ma(\pi)=1$ as indication 
of perfect competition and of full monopoly power, 
respectively.

Cash additivity implicitly assumes that the {\it num\'eraire} 
is traded and priced on a separate market -- as is normally 
the case with treasury bonds, but not necessarily so with 
other assets. Concerning positive homogeneity, this 
property is consistent with market imperfections such 
as the bid ask spread, but contrasts in general with fixed 
costs. In \cite{jouini_kallal_99} prices are described by a 
{\it convex} functional, thus not necessarily positively 
homogeneous%
\footnote{
However in \cite[Assumption 2.1]{jouini_kallal_99}
the assumption of monotonicity of prices is omitted.
}.
Indeed a number of our results still hold true even 
for convex price functions although the lack of positive
homogeneity cannot be reconciled with superhedging
duality established in Theorem \ref{th M}.

We also require that prices be free of arbitrage 
opportunities, a property which we define as%
\footnote{
See \cite{ECTH_2017} for a short discussion of alternative
definitions of arbitrage in an imperfect market. We do not 
adopt the pointwise definition of arbitrage suggested, e.g. 
in \cite{acciaio_et_al}, as this would implicitly correspond 
to assuming a form of rationality on economic agents even 
more extreme than probabilistic sophistication. 
}
\begin{equation}
\label{NA pos}
X\in\Xa,\ X>_*0
\qtext{imply}
\pi(X)>0.
\end{equation}

Of course, \eqref{NA pos} implies that $\pi(X)\ge0$ 
whenever $X\ge_*0$  while \eqref{NA mon} need not 
follow from \eqref{NA pos} if short selling is not 
permitted. We notice that the situation $X>_*\pi(X)>0$, 
though clearly exceptional, does not represent in our 
model an arbitrage opportunity because of the potential 
infeasibility of short positions in the {\it num\'eraire} 
asset. A firm experiencing difficulties in raising funds 
for its projects and competing with other firms in a 
similar position may offer abnormally high returns 
to induce investors to purchase its debt.

A functional satisfying all the preceding properties --
including \eqref{NA pos} -- will be called a {\it price 
function} and the corresponding set will be indicated 
with the symbol $\Pi(\Xa)$. We thus agree that market
prices are free of arbitrage by definition and we shall
avoid recalling this crucial property. At times, though,
it will be mathematically useful to consider pricing 
functionals for which cash additivity and/or the no 
arbitrage property \eqref{NA pos} may fail. These will 
be denoted by the symbol $\Pi_0(\Xa)$. To this end
we note that even if $\pi_0\in\Pi_0(\Xa)$ fails to be 
cash additive, it always has a cash additive part $\pi_0^a$, 
i.e. the functional
\begin{equation}
\label{pia}
\pi_0^a(X)
	=
\inf\big\{\pi_0(X+t)-t:t\in\R\qtext{such that} X+t\in\Xa\big\}
\qquad
X\in\Xa.
\end{equation}
It is routine to show that $\pi_0^a$ is the greatest 
cash additive element dominated by $\pi_0$. 

\begin{example}
In classical, continuous-time financial models, the 
discounted, final payoff of each investment takes 
the form
\begin{equation}
X^{w,\theta}
	=
w+\int\theta dS
\end{equation}
in which $w$ is the initial wealth invested, $S$,
the discounted price process, is a semimartingale 
on some filtered probability space and $\theta$ 
an element of a suitably defined set $\Theta$ of
admissible trading strategies. Then, 
$\pi(X^{w,\theta})=w$. Given the possibility
of shorting the num\'eraire, each investment
is naturally associated with a corresponding
self financing strategy with final payoff
$
X^{0,\theta}
	=
\int\theta dS
$.
The no arbitrage condition \eqref{NA pos},
restricted to self financing strategies, implies 
then that no such strategy may produce a 
strictly positive final payoff, i.e. that
\begin{equation}
\label{NA DS}
\big\{X^{0,\theta}:\theta\in\Theta\big\}
\cap 
L^0(\Omega,\A,P)_{++}
	=
\emp
\end{equation}
Compare with \cite[Definition 2.8(i)]
{delbaen_schachermayer}.
\end{example}

\section{Pricing Charges}
\label{sec pricing measures}

Associated with each price $\pi\in\Pi(\Xa)$ 
is the convex cone
\begin{equation}
\label{C(pi,A)}
\CapiA
	=
\big\{g\in\Fun{\Omega,\A}:
\lambda[X-\pi(X)]
\ge_* 
g
\text{ for some }
\lambda>0\text{ and }\ 
X\in\Xa
\big\}.
\end{equation}
and, more importantly, the collection of
{\it pricing probability charges}%
\footnote{
In \cite{ECTH_2017} we used the term pricing measure 
to define a positive charge dominated by $\pi$ without 
restricting it to be a {\it probability charge}. The focus 
on probability charges will be clear after Theorem 
\ref{th NFLVR}. A definition of the finitely additive 
integral and its properties may be found in 
\cite[III.2.17]{bible}.
}
\begin{equation}
\label{M(pi)}
\Meas_0(\pi,\A)
	=
\Big\{
m\in\Prob_{ba}(\A):
L(\Xa,\A)\subset L^1(\Omega,\A,m)
\text{ and }
\pi(X)\ge\int Xdm
\text{ for every }
X\in\Xa
\Big\}.
\end{equation}

The following, basic result illustrates the role of cash 
additivity and of the set $\Meas_0(\pi,\A)$. In particular,
property \eqref{attain} may be considered as our version 
of the classical superhedging duality Theorem%
\footnote{
The first part of this Theorem is mainly a consequence 
of the Hahn-Banach Theorem and it is true even if the 
price function $\pi_0$ is just convex. The superhedging 
formula \eqref{attain}, on the other hand, necessarily 
requires positive homogeneity.
}. 

\begin{theorem}
\label{th M}
For given $\pi_0\in\Pi_0(\Xa)$ the set $\Meas_0(\pi,\A)$ is 
non empty and each $m\in\Meas_0(\pi,\A)$ satisfies
\begin{equation}
\label{monotone m}
\int fdm
\ge
\int gdm
\qquad
f,g\in L(\Xa,\A),
f\ge_*g.
\end{equation}
Moreover, $\Meas_0(\pi_0,\A)=\Meas_0(\pi_0^a,\A)$ 
and
\begin{equation}
\label{attain}
\pi_0^a(X)
	=
\sup_{m\in\Meas_0(\pi_0,\A)}\int Xdm
\qquad
X\in\Xa\cap\B(\Omega,\A).
\end{equation}
Eventually, the set $\Meas_0(\pi_0,\A)$ is convex and 
compact in the weak$^*$ topology of $ba(\A)$ (i.e. the
topology induced by $\B(\Omega,\A)$ on $ba(\A)$).
\end{theorem}

\begin{proof}
We simply use Hahn-Banach and the integral 
representation of a positive linear functional
$\phi$ on the vector lattice $L(\Xa,\A)$ as
\begin{equation}
\label{rr}
\phi(f)
	=
\phi^\perp(f)+\int fdm_\phi
\qquad
f\in L(\Xa,\A)
\end{equation}
established in Theorem \ref{th rep}. In \eqref{rr} $m_\phi$ 
is a positive charge on $\A$ such that 
$L(\Xa,\A)\subset L^1(\Omega,\A,m_\phi)$; moreover, 
$m_\phi\in\Prob_{ba}(\A)$ if and only if $\phi(1)=1$.

We easily deduce from \eqref{NA mon} that the 
functional defined by
\begin{equation}
\label{superhedge}
\widebar\pi_0(g)
	=
\inf\big\{\lambda\pi_0(X):
\lambda>0,\ X\in\Xa,\ \lambda X\ge_*g\big\}
\qquad
g\in L(\Xa,\A)
\end{equation}
is an element of $\Pi_0(L(\Xa,\A))$ which extends 
$\pi_0$. By Hahn-Banach, we can find a linear 
functional $\phi$ on $L(\Xa,\A)$ such that 
$\phi
	\le
\widebar\pi_0$ 
and $\phi(1)=1$. Necessarily, $f\ge_*0$ implies $\phi(f)\ge0$ 
so that, by \eqref{ge affine}, $\phi$ is $\ge_*$ monotone and 
thus $m_\phi\in\Prob_{ba}(\A,\Neg_*)$. To show that 
$m_\phi\in\Meas_0(\pi_0,\A)$ observe that, by assumption, 
each $X\in\Xa$ admits $a\in\R$ such that $X\ge_*a$ so 
that $\phi^\perp(X)\ge\phi^\perp(a)=0$ and thus
$\pi_0(X)
	=
\widebar\pi_0(X)
	\ge
\phi(X)
	\ge
\int Xdm_\phi$. 
Suppose now that $m\in\Meas_0(\pi_0,\A)$ and that 
$f,g\in L(\Xa,\A)$ and $f\ge_*g$. Then, by \eqref{ge affine}, 
$\{f-g\le-\varepsilon\}\in\Neg_*$ for all $\varepsilon>0$ 
so that 
\begin{equation*}
\int(f-g)dm
=
\int_{\{f-g>-\varepsilon\}} (f-g)dm
\ge
-\varepsilon.
\end{equation*}
We deduce \eqref{monotone m} from $f,g\in L^1(\Omega,\A,m)$.

Concerning the claim 
$\Meas_0(\pi_0,\A)
	=
\Meas_0(\pi_0^a,\A)$, 
it is clear that the inequality $\pi_0^a\le\pi_0$ induces the 
inclusion $\Meas_0(\pi_0^a,\A)\subset\Meas_0(\pi_0,\A)$. 
Conversely, choose $m\in\Meas_0(\pi_0,\A)$. If $X\in\Xa$
and $t\in\R$ are such that $X+t\in\Xa$, then
\begin{align*}
\pi_0(X+t)-t
	\ge
\int (X+t)dm-t
	=
\int Xdm
\end{align*}
so that $m\in\Meas_0(\pi_0^a,\A)$.

The cash additive part $\widebar\pi_0^a$ of $\widebar\pi_0$,
obtained as in \eqref{pia}, is easily seen to be an extension 
of $\pi_0^a$ to $L(\Xa,\A)$. Of course, $\widebar\pi_0^a$ is the 
pointwise supremum of the linear functionals $\phi$ that 
it dominates so that \eqref{attain} follows if we show that 
$m_\phi\in\Prob_{ba}(\A)$ for all such $\phi$. But this is
clear since $\widebar\pi_0^a\ge\phi$ implies that $\phi$ is
positive on $L(\Xa,\A)$. Moreover, 
\begin{align*}
\widebar\pi_0^a(f)
	=
\widebar\pi_0^a(f+t)-t
	\ge
\phi(f+t)-t
	=
\phi(f)-t(1-\phi(1))
\qquad
t\in\R
\end{align*}
which contradicts the inequality $\widebar\pi_0^a\ge\phi$
unless $\phi(1)=1$ and thus $m_\phi\in\Prob_{ba}(\A)$. 

The last claim is an obvious implication of Tychonoff
theorem \cite[I.8.5]{bible}.
\end{proof}

Pricing probability charges closely correspond to the 
risk-neutral measures which are ubiquitous in the traditional
financial literature since the seminal paper of Harrison 
and Kreps \cite{harrison_kreps}. We only highlight that 
the existence of pricing charges and their properties 
are entirely {\it endogenous} here and do not depend 
on any special mathematical assumption -- and actually 
not even on the absence of arbitrage%
\footnote{
Nevertheless, Theorem \ref{th M} relies on property
\eqref{NA mon} which may be seen as a weaker form
of the no arbitrage property.
}. 
In traditional models,
the condition $\Meas_0(\pi,\A)\ne\emp$ is obtained via Riesz 
representation theorem (here replaced with Theorem \ref{th rep}) 
and requires an appropriate topological structure. Finite
additivity is a direct consequence of our minimal approach.
The view expressed by Bewley \cite[p. 516]{bewley} that
charges have ``no economic interpretation''
and the elements he offers in favour of the choice of the 
Mackey topology only make sense if a reference, countably 
additive measure is assumed to be given exogenously.
It is, in other words, a somewhat circular argument. We 
shall argue in section \ref{sec c.a.} that in a model
treating economic rationality as a primitive concept, the 
economic role of countable additivity is far from clear.

It is customary to interpret the expected value
$\int Xdm$
of the asset payoff as its {\it fundamental value}, 
although the value so obtained may vary significantly 
as $m$ ranges across the set of pricing charges.
The {\it intrinsic value} of the asset, computed as
\begin{equation}
\sup_{m\in\Meas_0(\pi,\A)}\int Xdm,
\end{equation}
corresponds to the most optimistic among such
evaluations. A {\it bubble} is customarily defined 
as the spread between the price of an asset and 
its intrinsic value, that is
\begin{equation}
\label{bubble}
\beta_\pi(X)
	\equiv
\pi(X)-\sup_{m\in\Meas_0(\pi,\A)}\int Xdm
\end{equation}
According to Theorem \ref{th M}, assets with 
bounded payoff admit no bubbles. At present,
however, we cannot exclude the extreme situation 
of a {\it pure bubble}, i.e. of an asset $X\in\Xa$ 
such that $X\ge_*0$,
\begin{equation}
\label{pure bubble}
\pi(X)>0
\qtext{but}
\sup_{m\in\Meas_0(\pi,\A)}\int Xdm=0.
\end{equation}
In Theorem \ref{th NFLVR} we obtain, among other things,
a full characterisation of pure bubbles.

At this level of generality we cannot conclude in favour
of the pricing formula popular in microstructure models 
and according to which asset prices are obtained by 
applying some spread to its fundamental value, such as
\begin{equation}
\label{mark-up}
\pi(X)
	=
[1+\alpha(X)]\int Xdm
\qtext{with}
\abs{\alpha(X)}<1
\qquad
X\in\Xa.
\end{equation}
We easily see that \eqref{mark-up} is a stronger condition 
than absence of pure bubbles. This formula will be discussed 
again in section \ref{sec competition}. 

We close remarking that each $m\in\Meas_0(\pi,\A)$
induces a linear functional on $L(\Xa,\A)$. The locally
convex linear topology induced by such functionals and
denoted by $\tau(\pi)$ is weaker than the classical
 weak topology.

\section{Market Completeness.}
\label{sec complete}

Competition among financial intermediaries may involve
existing assets and/or the launch of new financial claims. 
As a consequence it  may produce two different effects: 
\tiref a
a reduction of intermediation margins, and thus lower asset 
prices, and 
\tiref b
an enlargement of the set $\Xa$ of traded assets,
thus contributing to complete the markets. This short 
discussion justifies our interest for the set%
\footnote{
The set $\Ext_0(\pi,\A)$ is defined likewise but with $\Pi(\Xa)$
replaced with $\Pi_0(\Xa)$.
}
\begin{equation}
\label{Ext}
\Ext(\pi,\A)
	=
\big\{
\pi'\in\Pi\left(L(\Xa,\A)\right):
\restr{\pi'}{\Xa}\le\pi
\big\}.
\end{equation}
In this section we want to address the following question:
{\it under what conditions is it possible to extend the actual
markets to obtain an economy with complete financial 
markets without violating the no arbitrage principle?}
This translates into the mathematical condition 
$\Ext(\pi,\A)\ne\emp$ and, if $\pi'\in\Ext(\pi,\A)$, we 
speak of $(L(\Xa,\A),\ge_*,\pi')$ as an $\A$-completion 
of $(\Xa,\ge_*,\pi)$.

The first papers to focus on the extension property 
of financial prices were those by Harrison and Kreps 
\cite{harrison_kreps} and Kreps \cite{kreps} (see 
also \cite[Theorem 8.1]{MAFI_2008}), although this 
aspect has later been somehow neglected in the 
following literature. In their approach the completion
property is related to (and in fact in most cases
equivalent to) the concept of viability that we discuss
later. Recalling the discussion in subsection 
\ref{subsec complete}, we remark that our definition 
of $\A$-completion does not preclude that markets 
may be in principle further extended, e.g. by passing 
to a larger $\sigma$ algebra than $\A$. However, 
given that $\A$ is quite arbitrary, our results carry 
over very simply.

We obtain the following complete characterisation for 
the case of cash additive $\A$-completions.

\begin{theorem}
\label{th NFLVR}
For a market $(\Xa,\ge_*,\pi)$ the following properties 
are mutually equivalent: 
\begin{enumerate}[(a).]
\item\label{it wNFL}
$\pi$ satisfies the condition
\begin{equation}
\label{wNFLVR}
\cl[\tau(\pi)]{\CapiA}
\cap
\{f\in\Fun{\Omega,\A}:f>_*0\}
	=
\emp;
\end{equation}
\item\label{it NFL}
$\pi$ satisfies the condition
\begin{equation}
\label{NFLVR}
\cl[u]{\CapiA}
\cap
\{f\in\Fun{\Omega,\A}:f>_*0\}
	=
\emp;
\end{equation}
\item\label{it compl}
the market $(\Xa,\ge_*,\pi)$ admits an $\A$-completion;
\item\label{it NPB}
the set $\Meas_0(\pi,\A)$ of pricing probability charges 
satisfies the condition
\begin{equation}
\label{NPB}
\sup_{m\in\Meas_0(\pi,\A)}\int (f\wedge 1)dm
	>
0
\qtext{for all}
f\in\Fun{\Omega,\A}
\text{ such that }
f>_*0.
\end{equation}
\end{enumerate}
\end{theorem}

\begin{proof}
The implication \imply{it wNFL}{it NFL} is immediate
in view of the fact that $\tau(\pi)$ is a topology
weaker than the topology of uniform distance.
Assume that \eqref{NFLVR} holds and define the functional
\begin{equation}
\rho(f)
	=
\inf\big\{\lambda\pi(X)-a:
a\in\R,
\lambda>0,
X\in\Xa
\text{ such that }
\lambda X\ge_* f+a\big\}
\qquad
f\in L(\Xa,\A).
\end{equation}
Property \eqref{ge affine} and $\ge_*$ monotonicity 
of $\pi$ imply that $\rho\in\Ext_0(\pi,\A)$. Moreover, 
it is easily seen that 
\begin{equation}
\label{rhoad}
\rho(f+x)
	=
\rho(f)+x
\qquad
f\in L(\Xa,\A), 
x\in\R
\end{equation} 
and thus that $\rho$ is cash additive. To prove that 
$\rho
	\in
\Pi\big(L(\Xa,\A)\big)$, 
fix $f\in L(\Xa,\A)$
such that $f_1=f\wedge 1>_*0$. In search of a contradiction, 
suppose that $\rho(f_1)\le 0$. Then for each $n\in\N$ 
there exist $a_n\in\R$, $\lambda_n>0$ and $X_n\in\Xa$ 
such that 
$\lambda_nX_n
	\ge_* 
f_1+a_n$ 
but 
$\lambda_n\pi(X_n)
	<
2^{-n}+a_n$. 
This clearly implies
\begin{equation}
\lambda_n[X_n-\pi(X_n)]
	\ge_*
f_1-2^{-n}
\end{equation}
and $f_1\in\cl[u]{\CapiA}$, contradicting
\eqref{NFLVR}. It follows that $\rho(f_1)>0$ 
and that \imply{it NFL}{it compl}.

Choose $\rho\in\Ext(\pi,\A)$ and let $f_1$ be
defined as above. Consider the linear 
functional 
\begin{equation}
\hat\phi(x+bf_1)
	=
x+b\rho(f_1)
\qquad
x,b\in\R
\end{equation}
defined on the linear subspace $L_0\subset L(\Xa,\A)$ 
spanned by $\{1,f_1\}$. Given that $\rho$ satisfies 
\eqref{rhoad}, $\hat\phi$ is dominated by $\rho$ 
on $L_0$ so that we can find an extension $\phi$ of 
$\hat\phi$ to the whole of $L(\Xa,\A)$ still dominated 
by $\rho$. As in Theorem \ref{th M}, given that 
$\phi$ is a positive linear functional on a vector 
lattice of functions, we obtain the representation 
\eqref{phi rep} with 
$m_\phi\in\Meas_0(\rho,\A)
	\subset
\Meas_0(\pi,\A)$. 
Observe that $f_1\sim_*f_1^+$ and that 
$f_1^+\in\B(\Omega,\A)$.
We deduce from Theorem \ref{th M},
\begin{align*}
0
	<
\rho(f_1)
	=
\phi(f_1)
	=
\phi(f_1^+)
	=
\int f_1^+dm_\phi
	=
\int f_1dm_\phi.
\end{align*}
Thus, \imply{it compl}{it NPB}.

Let now
$f
	\in
\cl[\tau(\pi)]{\CapiA}$. For each $m\in\Meas_0(\pi,\A)$
and $n\in\N$ there exist $\lambda^n>0$,
$X^n\in\Xa$ and $h_n\in L(\Xa,\A)$ such that 
$\lambda^n(X^n-\pi(X^n))\ge_*h_n$
and
$\int(f-h_n)dm\le2^{-n}$.
But then,
\begin{align*}
\int fdm
\le
2^{-n}+\lambda^n\Big[\int X^n dm-\pi(X^n)\Big]
\le
2^{-n}.
\end{align*}
Under \iref{it NPB} this excludes that $f>_*0$ and
proves the implication \imply{it NPB}{it wNFL}.
\end{proof}

It is immediate to recognize the close relationship linking 
the condition \eqref{NFLVR} to the {\it No-Free-Lunch-with%
-Vanishing-Risk} principle formulated long ago by Delbaen 
and Schachermayer \cite{delbaen_schachermayer} in a 
highly influential paper. Since then, this condition, despite 
its unclear economic interpretation, has been unanimously 
accepted as the most convenient mathematical restatement 
of the no arbitrage principle. Leaving aside the dissimilarity 
in the set-up adopted, the major difference between the 
two conditions lies in the interpretation. In fact, due to 
the restrictions to trade considered here, the elements of 
the form $X-\pi(X)$, with $X\in\Xa$, need not correspond 
to the payoff of any feasible trading strategy so that 
\eqref{NFLVR} cannot be interpreted as a mathematical 
reformulation of the no arbitrage principle. Rather, the 
set $\cl[u]\CapiA$ represents those potential claims that 
cannot be assigned a strictly positive price by {\it any} 
extension of the actual 
price function. Thus, Theorem \ref{th NFLVR} characterizes 
\eqref{NFLVR} as a condition necessary and sufficient for 
financial markets to admit a no arbitrage $\A$-completion. 
Notice that a strictly positive extension may still exist 
even when \eqref{NFLVR} fails. In this case, however, 
it cannot be cash additive. 

In the light of the discussion following \eqref{pure bubble}, 
condition \eqref{NPB} corresponds to a {\it No-Pure-Bubble} 
({\it NPB}) condition, although it does not exclude more general 
bubbles defined as in \eqref{bubble}. In other terms, with 
complete financial markets there cannot exist pure bubbles,
which are instead possible with incomplete markets%
\footnote{
The fact that completeness of financial markets may 
change the structure of asset bubbles has already 
been noted by Jarrow, Protter and Shimbo 
\cite{jarrow_protter_shimbo}.
}.

Given the arbitrariness of the $\sigma$ algebra $\A$,
one may consider the possibility of further extending
the market from $L(\Xa,\A)$ to $L(\Xa',\A')$ where
$\Xa'=L(\Xa,\A)$ and $\A\subset\A'$. If $\pi'\in\Ext(\pi,\A)$,
it is then immediate from Theorem \ref{th NFLVR}
that $(\Xa',\ge_*,\pi')$ admits an $\A'$-completion
if and only if 
\begin{equation}
\label{NFLVR1}
\cl[u]{\mathscr C(\pi',\A')}
\cap
\{f\in\Fun{\Omega,\A'}:f>_*0\}
=
\emp.
\end{equation}
Notice that each $\pi''\in\Ext(\pi',\A')$ satisfies
the supermartingale-like inequality
$\restr{\pi''}{\Xa'}\le\pi'$

Incidentally we remark that,  from the equality 
$\Meas_0(\pi_0,\A)
	=
\Meas_0(\pi^a_0,\A)$, 
it follows that $\pi_0\in\Pi_0(\Xa)$ satisfies \eqref{NFLVR} 
if and only if so does $\pi_0^a$. In fact,

\begin{lemma}
\label{lemma aNFLVR}
Let $\pi\in\Pi_0(\Xa)$. Then: 
(a)
$
\CapiA
\subset
\Ca(\pi^a,\A)
\subset
\cl[u]{\CapiA}$
and
(b)
for every $X\in\Xa$,
$\pi^a(X)\le0$ if and only if $X\in\cl[u]{\CapiA}$.
Therefore, $\pi^a\in\Pi(\Xa)$ if and only if
\begin{equation}
\cl[u]{\CapiA}
\cap
\{X\in\Xa:X>_*0\}
	=
\emp.
\end{equation}
\end{lemma}

\begin{proof}
\tiref a.
For each $X\in\Xa$ it is obvious that $X-\pi(X)\le X-\pi^a(X)$. 
However, $X-\pi^a(X)$ is the limit in the uniform topology
as $t\to+\infty$, of
$
X+t-\pi(X+t)
	\in
\CapiA$.
\tiref b.
$X\in\Xa$ and $\pi^a(X)\le0$ imply that for each $n\in\N$ 
and for $t_n>0$ sufficiently large
\begin{align*}
X
	\le 
X-\pi^a(X)
	\le
2^{-n}+[X+t_n-\pi(X+t_n)]
\end{align*}
so that $X\in\cl[u]{\CapiA}$. Viceversa, if
$X\le2^{-n}+\lambda_n[X_n-\pi(X_n)]$ for some 
$X_n\in\Xa$ and $\lambda_n\ge0$, then, moving
$\lambda_n\pi(X_n)$ to the left hand side if positive
and using cash additivity, we conclude
$\pi^a(X)
	\le 
2^{-n}+\lambda_n[\pi^a(X_n)-\pi(X_n)]
	\le
2^{-n}$.
\end{proof}

To highlight the role of competition in financial markets,
consider two pricing functions $\pi,\pi'\in\Pi(\Xa)$. If 
$\pi\le\pi'$ then $\Ca(\pi',\A)\subset\CapiA$. Thus, lower 
financial prices are less likely to satisfy \eqref{NFLVR} 
and thus to admit an extension to a complete financial 
market free of arbitrage. Competition among market 
makers, producing lower spreads, may thus have two 
contrasting effects on economic welfare. On the one 
side it reduces the well known deadweight loss implicit 
in monopolistic pricing while, on the other, it imposes 
a limitation to financial innovation and its benefits in 
terms of the optimal allocation of risk. It may be 
conjectured that fully competitive pricing, i.e. the 
pricing of assets by their fundamental value, may not 
be compatible with the extension property discussed 
here. We investigate this issue in the following Section 
\ref{sec competition}. 

%

\section{Viability}
\label{sec viable}
Following Harrison and Kreps \cite{harrison_kreps} and
Kreps \cite{kreps}, several authors in financial economics
have characterized the notion of {\it viability}, i.e. the 
property that price functions support optimal choice by 
individuals with monotone, convex and continuous 
preferences. Continuity of preferences is a key property
in proving the existence of economic equilibrium (see the 
survey by Mas-Colell and Zame \cite{mas-colell_zame}).
In addition, in most models this property provides the
necessary justification for formulating the separating 
condition \eqref{NFLVR} in terms of the closure of 
$\CapiA$.

In our model, viability induces a conclusion much
stronger than condition \eqref{NFLVR} which is in 
fact equivalent to  a local version of it. Moreover, 
the restrictions to trade considered, particularly 
the constraint that prevents shorting the 
{\it num\'eraire}, drive a wedge between our 
approach and other papers in this literature (but 
one should refer to Jouini and Kallal \cite{jouini_kallal_99}
for a comparison). 

Let us consider the space $\mathbb X=\R\times\Fun{\Omega,\A}$,
endowed with the topology of uniform convergence,
as a description of agents consumption space at two
different instants of time. Agents with no initial 
endowment are described by a strict preference
$\succ$ (i.e. a transitive and non reflexive binary 
relation) which induce family of preferred sets
\begin{equation}
\label{W}
\W(c,h)
	=
\{(c',h')\in\mathbb X:
(c',h')\succ(c,h)
\}
\qquad
(c,h)\in\mathbb X
\end{equation}
and by the budget set, defined as
\begin{equation}
\label{budget}
\mathcal B(\pi)
	=
\big\{(c,h)\in\mathbb X:
c+\lambda\pi(X)\le 0
\text{ and }
\lambda X\ge_*h
\text{ for some }
\lambda>0
\text{ and }	
X\in\Xa
\big\}.
\end{equation}

We can then define viability in formal terms:
\begin{definition}
The price function $\pi$ is viable relatively to the
strict preference relation $\succ$ if
\begin{equation}
\label{viable}
\mathcal B(\pi)\cap \W(0,0)
	=
\emp.
\end{equation}
\end{definition}

Viability has special importance when preferences 
are monotonic, convex and continuous. In order to 
define these properties, we extend $>_*$ to 
$\mathbb X$ 
by writing
\begin{equation}
(c,h)
	>_*
0
\qtext{whenever}
c\ge0,\ h\ge_*0
\text{ and }
c+h>_*0,
\end{equation}
and define the convex cone (with the origin deleted)
\begin{equation}
\label{K}
\K
	=
\{(c,h)\in\mathbb X:
(c,h)>_*0
\}.
\end{equation}
$\K$ plays here the same role as in \cite{kreps}. We
define now the class of admissible preferences.

\begin{definition}
\label{def class A}
A strict preference $\succ$ on $\R\times\Fun{\Omega,\A}$
is of class $\mathbf A$ if it satisfies the following properties:
\begin{subequations}
\label{adm}
\begin{equation}
\label{adm monotone}
(c,h)+\K
\subset 
\W(c,h)
\qquad
(c,h)\in\mathbb X
\end{equation}
\begin{equation}
\label{adm convex}
\W(c,h)
\qtext{is convex}
(c,h)\in\mathbb X
\end{equation}
\begin{equation}
\label{adm cts}
\big(\forall k\in\K)
\big(\exists a_k>0\big):
\quad
\{a k:a>a_k\}\subset\Int\big(\W(0,0)\big)%
\footnote{
$\Int(A)$ denotes the interior of the set $A$. Recall
that $\mathbb X$ is endowed with the uniform topology.
}.
\end{equation}
\end{subequations}
\end{definition}

Condition \eqref{adm monotone} is a monotonicity 
property while \eqref{adm cts} is a weak form of 
semicontinuity. We can now define two distinct 
notions of viability as follows:

\begin{definition}
\label{def adm}
The price function $\pi$ is 
(a)
viable if it is $\succ$-viable for some $\succ$ of class 
$\mathbf A$,
(b)
$*$-viable if either
\begin{enumerate}[(i)]
\item
it is viable and 
$(\varepsilon,-\varepsilon)\in\mathcal B(\pi)$ 
for some $\varepsilon>0$ or
\item
it is $\succ$-viable for some $\succ$ of class $\mathbf A$
which satisfies the additional condition
\begin{equation}
\label{adm indiff}
\W(-t,t)=\W(0,0)
\qtext{for some}
t>0.
\end{equation}
\end{enumerate}
\end{definition}

The literature has rarely considered the need to reinforce
the notion of viability as we do here introducing $*$-viability.
The reason is that the two definitions coincide whenever
either the {\it num\'eraire} is traded freely or preferences
are defined over net trades. And a combination of either one
of these two features appears in virtually all contributions.
It is thus of some importance to remark that outside of the 
narrow, traditional approach viability may need to be 
reformulated.

Another new feature of our model is the focus on a {\it local} 
version of the above properties. If $k\in\K$, a strict preference 
is of class $\mathbf A_k$ if it satisfies the above properties 
\eqref{adm} in restriction to the convex cone 
$\{ak+(b,0):a,b\in\R_+\}$ 
(rather than the whole of $\K$). If for any $k\in\K$ the price 
function $\pi$ is viable for some preference of class $\mathbf A_k$, 
we speak of $\pi$ as locally viable (resp. locally $*$-viable).

\begin{theorem}
\label{th viable}
(i).
A price function $\pi$ satisfies \eqref{NFLVR} if and
only if it is locally $*$-viable.
(ii)
A price function $\pi$ is $*$-viable if and
 only if it admits
$m\in\Meas_0(\pi,\A)$ such that
\begin{equation}
\label{M+}
\int (f\wedge1)dm>0
\qquad
f\in L(\Xa,\A),\ f>_*0.
\end{equation}
\end{theorem}

\begin{proof}
\tiref i.
Endow the space 
$\mathbb X_b
	=
\R\times\B(\Omega,\A)$ 
with the norm 
$\norm{(c,h)}
	=
\abs c+\norm h_{\B(\Omega,\A)}$
and assume that $\pi$ is locally $*$-viable. Fix 
$f>_*0$, write $k=(0,f\wedge1)$ and choose $\succ$ 
to be of class $\mathbf A_k$ and such that $\pi$ is 
$\succ$-viable. Then, 
$\W(0,0)\cap\mathbb X_b$ 
and 
$\mathcal B(\pi)\cap\mathbb X_b$ 
are disjoint, convex sets, the former has non empty 
interior (by \eqref{adm cts}) and contains a convex 
cone (by \eqref{adm monotone}) while the latter 
contains the origin. There exists then, \cite[V.1.12]{bible}
a non trivial, continuous, linear functional $\Phi$
such that 
\begin{equation}
\label{separ}
\inf_{w\in\W(0,0)\cap\mathbb X_b }\Phi(w)
\ge
0
\ge
\sup_{b\in\mathcal B(\pi)\cap\mathbb X_b}\Phi(b).
\end{equation}
Thus $\Phi((-1,1))\le0$. If
$(\varepsilon,-\varepsilon)
\in
\mathcal B(\pi)$
for some $\varepsilon>0$ then $\Phi((-1,1))=0$. 
Otherwise, $*$-viability implies that $\succ$ may 
be chosen so as to satisfy \eqref{adm indiff}. Fix 
then $t$ as in \eqref{adm indiff} and let $c>0$ be 
arbitrary. We have, $(c-t,t)\succ(-t,t)$ by 
\eqref{adm monotone} so that 
$(c-t,t)
	\in
\W(0,0)$. 
Continuity of $\Phi$ induces the conclusion that 
$\Phi((-t,t))\ge 0$ and again we obtain $\Phi((-1,1))=0$. 
Given that $\Phi(\Int(\W(0,0)))$ is an open interval 
in $\R_+$ \cite[p. 113]{conway} and that 
$ak\in\Int(\W(0,0))$ for some $a>0$ we conclude 
that $\Phi(k)>0$. By the representation of continuous
linear functionals we conclude that $\Phi$ may be 
written in the form
\begin{equation}
\Phi(c,h)
	=
xc
+
\int hd\mu
\end{equation}
for some $x\in\R$ and some charge $\mu\in ba(\A)$. 
The inclusion
$-\mathbb X_{b,+}
	\subset
\mathcal B(\pi)$,
implies that $x\ge0$, $\mu\ge0$ and $\mu\ne0$
(because $\Phi(k)=\int(f\wedge 1)d\mu>0$); in
addition $\Phi(1,-1)=0$ leads to $x=\norm\mu$. 
By normalization we obtain that
\begin{equation}
\label{Phi}
\Phi(c,h)
=
\int(c+h)dm
\qquad
(c,h)\in\mathbb X_b
\end{equation} 
for some $m\in\Prob_{ba}(\A)$. The inclusion
$(0,\set N)\in\mathcal B(\pi)$, valid for all 
$N\in\Neg_*$, implies $m\in\Prob_{ba}(\A,\Neg_*)$. 
Eventually, if $X\in\Xa$, there exists $a\in\R$ 
such that $X\sim_*X\vee a$ so that, for any 
$n\in\N$ we have the inclusion 
$(-\pi(X),(X\vee a)\wedge 2^{-n})
\in
\mathcal B(\pi)\cap\mathbb X_b$,
from which we deduce
\begin{equation}
\label{domin}
\int(X\wedge 2^{-n})dm
	=
\int[(X\vee a)\wedge 2^{-n}]dm
	=
\Phi(-\pi(X),(X\vee a)\wedge 2^{-n})+\pi(X)
	\le
\pi(X).
\end{equation}
This, being true for all $X\in\Xa$ and $n\in\N$, implies 
via \cite[III.3.6]{bible} that $\Xa\subset L^1(\Omega,\A,m)$ 
and, eventually, that $m\in\Meas_0(\pi,\A)$. Summing up,
 for each $f>_*0$ there exists $m\in\Meas_0(\pi,\A)$ such 
 that $\int(f\wedge 1)dm>0$, a condition equivalent to 
 \eqref{NFLVR} by Theorem \ref{th NFLVR}.

Viceversa, assume \eqref{NFLVR}, fix $k_0=(c_0,h_0)\in\K$ 
and choose $m_0\in\Meas_0(\pi,\A)$ such that 
$\int(c_0+ h_0)dm_0>0$%
\footnote{
The existence of such $m$ is trivial if $c_0>0$ or else 
follows from \eqref{NPB}. 
}. 
Define a strict preference $\succ$ by letting
$(c,h)\succ(c',h')$ whenever $h-h'\ge_*a$
for some $a\in\R$ and
\begin{equation}
\label{pref}
(c-c')+\lim_k\int\big[(h-h')\wedge 2^k\big]dm_0
>
0.
\end{equation}
Properties \eqref{adm} and \eqref{adm indiff} are 
easily seen to hold true so that $\succ$ is of class 
$\mathbf A_{k_0}$. This, being true for all
$k_0\in\K$, proves that $\pi$ is locally $*$-viable.

\tiref{ii}.
If $\pi$ is $*$-viable, it is necessarily locally $*$-viable.
The proof is identical to that of claim \tiref i with the
only difference that $\Int(\W(0,0))$ contains now 
appropriate multiples of {\it each} $k\in\K$. Thus, if 
$\Phi$ is the continuous linear functional introduced 
in \eqref{separ} we conclude that $\Phi(k)>0$ for all 
$k\in\K$. Proceeding exactly as above, we obtain 
$m\in\Meas_0(\pi,A)$ that satisfies \eqref{M+}. 
Conversely, if $m\in\Meas_0(\pi,A)$ satisfies such
properties, the strict preference defined in
\eqref{pref} is necessarily $*$-viable.
\end{proof}

As is clear from the proof, $*$-viability is required 
to conclude that the set function $m$ obtained by 
normalizing $\Phi$ is a {\it probability} charge.
Assuming simple viability, we would in fact obtain 
an element of $ba(\A)_+$ which is dominated by 
$\pi$ on $\Xa$ and this is weaker than \eqref{NFLVR}.

The main finding of Theorem \ref{th viable} is that 
full viability of financial prices is a much more 
restrictive condition than \eqref{NFLVR} and is 
in fact equivalent to the existence of 
$m\in\Meas_0(\pi,\A)$ that satisfies \eqref{M+}. 
This point will be discussed at length in the following 
section \ref{sec competition} in which we provide 
examples in which such set functions may not exist%
\footnote{
Kreps \cite[Example 3, p. 23]{kreps} has documented 
other examples of a topological vector space admitting 
no continuous, strictly positive linear functional and in 
which, as a consequence, absence of arbitrage is a
much weaker condition than admissibility,  in line with
the conclusions of Theorem \ref{th viable} above. 
}.
In these cases, the general notion of viability as
defined above is not useful. Although this may 
appear at first as a limitation, the next result 
suggests that in fact it may not be so provided
we enlarge the set of preferences that we
consider as admissible\footnote{
This result compares with
\cite[Theorem 1]{loewenstein_willard_et}.
}.

\begin{theorem}
The price function $\pi$ satisfies \eqref{NFLVR} if
and only if it is viable for some strict preference
$\succ$ which satisfies \eqref{adm monotone},
\eqref{adm cts}%
\footnote{
Properties \eqref{adm monotone} and \eqref{adm cts}
are unduly restrictive and a local version may be used
instead. We may, in other words, require that for each
$k\in\K$ there exists a strict preference which is of
class $\mathbf A_k$ (save for convexity) and such that
$\pi$ is $\succ$-viable. This more general formulation
would however require a quite involved statement.
} 
as well as either
\begin{subequations}
\begin{equation}
\label{net}
(c,h)\in\W(0,0)
\qiff
(c-t,h+t)\in\W(0,0)
\qquad
t>0
\qtext{or}
\end{equation}
\begin{equation}
\label{borrow}
(-t,t)\in\mathcal B(\pi)
\qquad
t>0.
\end{equation}
\end{subequations}
\end{theorem}

\begin{proof}
Let $\succ$ meet the conditions of the claim and
choose $X\in\Xa$. If $\pi(X)\le0$ then 
$(0,X-\pi(X))\in\mathcal B(\pi)$; if on the other hand 
$\pi(X)>0$ then the same conclusion follows from either
It is thus clear that if $\pi$ is $\succ$-viable then necessarily 
$\{X-\pi(X):X\in\Xa\}$ and
$\mathcal H
	=
\{f\in\Fun{\Omega,\A}:(0,f)\in\W(0,0)\}$
are disjoint sets. By \eqref{adm monotone} the same 
conclusion remains valid upon replacing the first set 
with $\CapiA$. For each $f\in\Fun{\Omega,\A}$
such that $f>_*0$ (and thus $f\in\mathcal H$)
we deduce from \eqref{adm cts} that
$f\notin\cl[u]{\CapiA}$. This shows that
\eqref{NFLVR} holds. If, conversely, $\pi$ satisfies
\eqref{NFLVR}, then one may define
$(c,h)\succ(c',h')$ simply by letting
\begin{equation}
\pi(h)+c>\pi(h')+c'.
\end{equation}
It is immediate that this is a strict preference and
that it satisfies \eqref{adm monotone},
\eqref{adm cts} as well as \eqref{net}
while it is obvious that $\pi$ is $\succ$-viable.
\end{proof}

Thus, \eqref{NFLVR} i.e. local $*$-viability, are
indeed perfectly sensible conditions from the
point of view of market equilibrium, at least
under either \eqref{net} or \eqref{borrow} and 
as long as we do not  insist on convexity of 
preferences.

\section{Competitive Complete Markets.}
\label{sec competition}

In this section we investigate the conditions under which
the set $\Meas_0(\pi,\A)$ contains a strictly positive 
element, i.e. some $m$ satisfying \eqref{M+}. We
denote the corresponding set with the symbol
$\Meas(\pi,\A)$%
\footnote{
In \cite{riedel} a set function satisfying \eqref{M+} is 
said to have {\it full support} and the emergence of 
measures with full support follows easily from the 
assumption that $\Omega$ is a complete, separable 
metric space and $\Xa$ consists of continuous functions 
defined thereon. 
Notice that for a pricing probability charge to meet
\eqref{M+} it is necessary that it be strictly positive 
on each non negligible set. This condition is, however,
not sufficient. In fact, although $f>_*0$ implies 
$\{f>0\}\notin\Neg_*$, the fact that $\Neg_*$ is
not closed with respect to countable unions does
not rule out that $\{f>\varepsilon\}\in\Neg_*$ for
all $\varepsilon>0$ which contradicts \eqref{M+}.
}.
By Theorem \ref{th viable}, the condition 
$\Meas(\pi,\A)\ne\emp$ is equivalent to the price 
function $\pi$ being viable. More importantly, if 
$m\in\Meas(\pi,\A)$ then each asset with payoff in 
$L(\Xa,\A)$ may be priced by its fundamental value 
and this price rule $\rho$ would be free of arbitrage
and linear. In other words, $\rho$ would provide an 
$\A$-extension of $\pi$ which is fully competitive,
in symbols $\rho\in\Ext(\pi,\A)$ with $\ma(\rho)=0$
(recall the definition \eqref{market power} of the 
market power index $\ma$).

This is of course, an extreme situation. The question 
we want to address is rather: {\it given a market, 
$(\Xa,\ge_*,\pi)$, is it possible to find an $\A$-completion
that permits some degree of competitiveness?}. In other 
words, we look for $\rho\in\Ext(\pi,\A)$ satisfying the 
{\it No-Full-Monopoly} ({\it NFM}) condition
\begin{equation}
\label{NFM}
\ma(\rho)<1.
\end{equation}
As we shall see, this condition has in fact far reaching 
implications.

\begin{theorem}
\label{th NFM}
A market $(\Xa,\ge_*,\pi)$ satisfies 
$\Meas(\pi,\A)\ne\emp$ 
if and only if it admits a NFM $\A$-completion. 
\end{theorem}

\begin{proof}
Necessity is immediate. If $m\in\Meas(\pi,\A)$, define 
$\rho\in\Ext_0(\pi,\A)$ by
\begin{equation}
\rho(f)
	=
\int fdm
\qquad
f\in\ L(\Xa,\A).
\end{equation}
Then, $\rho$ is cash additive and $\ma(\rho)=0$, by 
linearity. Conversely, assume that $\rho\in\Ext(\pi,\A)$ 
satisfies \eqref{NFM}. For each $n\in\N$ define the sets
\begin{equation}
\Bor_0
=
\{b\in\B(\Xa,\A,\Neg_*):
1\ge_* b>_*0\}
\qtext{and}
\Bor_n
	=
\big\{b\in\Bor_0:
\rho(b)>2/[n(1-\ma(\rho)]\big\}
\end{equation}
and let $\co(\Bor_n)$ be the convex hull of $\Bor_n$. 
Notice that 
$\Bor_0
	=
\bigcup_n\Bor_n$,
because $\rho\in\Ext(\pi,\A)$. If
$f
	=
\sum_{i=1}^Nw_ib_i\in\co(\Bor_n)$ 
then,
\begin{align*}
\rho(f)
	\ge
\big(1-\ma(\rho)\big)\sum_{i=1}^Nw_i\rho(b_i)
	\ge
2/n.
\end{align*}
In view of the properties of $\Meas_0(\pi,\A)$ proved in 
Theorem \ref{th M}, we can apply Sion minimax 
Theorem \cite[Corollary 3.3]{sion} and obtain from 
\eqref{attain} 
\begin{align*}
\inf_{f\in\co(\Bor_n)}\rho(f)
	=
\inf_{f\in\co(\Bor_n)}\sup_{\mu\in\Meas_0(\rho,\A)}
\int fd\mu
	=
\sup_{m\in\Meas_0(\rho,\A)}\inf_{f\in\co(\Bor_n)}
\int fd\mu.
\end{align*}
Therefore, for each $n\in\N$ there exists 
$\mu_n\in\Meas_0(\rho,\A)$ such that 
$\inf_{f\in\co(\Bor_n)}\int fd\mu_n
	>
1/n$. 
Define
$m
	=
\sum_n2^{-n}\mu_n
$.
Then,
$m
	\in
\Meas_0(\rho,\A)
	\subset
\Meas_0(\pi,\A)$
and, as a consequence, $L(\Xa,\A)\subset L^1(\Omega,\A,m)$. But
then, if $f>_*0$ we conclude that 
$f\wedge1\in\Bor_0\subset L(\Xa,\A)$,
that $f\wedge1>_*0$ (by \eqref{ge trunc}) and that
$
\int(f\wedge 1)dm
	>
0$
so that $m\in\Meas(\pi,\A)$.
\end{proof}

What the preceding Theorem \ref{th NFM} asserts in words 
is that if an $\A$-complete, arbitrage free market is 
possible under limited market power, it is then possible 
under perfect competition -- i.e. with assets priced by 
their fundamental value. This conclusion does not exclude, 
however, the somewhat paradoxical situation in which 
the only possibility to complete the markets is by admitting 
unlimited market power by financial intermediaries. 
As noted in the introduction, this situation describes the terms 
of a potential conflict between the effort of regulating 
the market power of intermediaries and the support to 
a process of financial innovation that does not disrupt 
market stability by introducing arbitrage opportunities. 

We stress that the conclusions of Theorem \ref{th NFM}
do crucially depend on the intervening $\sigma$ algebra 
$\A$ and that the {\it NFM} condition is the less likely to hold
the finer is $\A$. This remark suggests that a limited
or null market power may not be possible as the degree
of complexity of the financial market, as embodied in
the width of $\A$, ranks high. Actually, given the
requirement $\Neg_*\subset\A$, even the smallest 
possible extension of markets as we defined it may 
not admit any extension satisfying \eqref{NFM}.

Consider, e.g., the case in which an uncountable 
family of possible, alternative scenarios is given. 
In mathematical terms we can model this situation 
via an uncountable, pairwise disjoint collection 
$\{A_\alpha:\alpha\in\mathfrak A\}$ 
of non negligible sets in $\A$. Then, if 
$\rho\in\Ext(\pi,\A)$ and letting 
$f_\alpha
	=
\set{A_\alpha}$, 
it must be that $\rho(f_\alpha)>0$ for each 
$\alpha\in\mathfrak A$.
Thus, upon choosing appropriately $\delta>0$ and 
$\alpha_1,\alpha_2,\ldots\in\mathfrak A$, we obtain
\begin{equation}
\inf_n\rho(f_{\alpha_n})
	>
\delta
\end{equation}
and, consequently,
$\sum_{1\le n\le N}(1/N)\rho(f_{\alpha_n})
	>
\delta$.
On the other hand, by disjointness,
\begin{align*}
\rho\Big(\tfrac1N\sum_{1\le n\le N}f_{\alpha_n}\Big)
	\le
\rho\Big(\tfrac1N\sup_{1\le n\le N}f_{\alpha_n}\Big)
	\le
1/N.
\end{align*}
Thus, in the case under consideration $\ma(\rho)=1$. 
Under the classical probabilistic assumptions, the 
existence of the collection
$\{A_\alpha:\alpha\in\mathfrak A\}$ 
described above is not possible%
\footnote{
In the theory of Boolean algebras the condition that 
no uncountable, pairwise disjoint collection of non 
empty sets may be given, is known as the {\it countable 
chain} (CC) condition and was first formulated by Maharam 
\cite{maharam}. See the comments in \cite{JMAA_2019}.
}.

A stronger form of the preceding example can be
proved.

\begin{lemma}
\label{lemma m=1}
Let the $\sigma$ algebra $\A$ admit an uncountable
collection of non negligible sets, the intersection of 
any two of which is negligible. Then, any extension 
$\rho\in\Ext(\pi,\A)$ is necessarily such that
$\ma(\rho)=1$.
\end{lemma}

\begin{proof}
Using exactly the same notation of the example
above, we obtain that $\rho(f_{\alpha_n})>\delta$
as before while 
$f_{\alpha_n}\wedge f_{\alpha_m}
\sim_*0$
when $n\ne m$. We obtain
\begin{equation}
f_{\alpha_1}\wedge\sum_{n=2}^Nf_{\alpha_n}
\le
\Big(f_{\alpha_1}\wedge\sum_{n=3}^Nf_{\alpha_n}\Big)
+
(f_{\alpha_1}\wedge f_{\alpha_2})
\sim_*
f_{\alpha_1}\wedge\sum_{n=3}^Nf_{\alpha_n}
\end{equation}
i.e., iterating this procedure, 
$f_{\alpha_1}\wedge\sum_{n=2}^Nf_{\alpha_n}
	\sim_*
(f_{\alpha_1}\wedge f_{\alpha_N})
	\sim_*
0$.
But then,
\begin{align}
\sum_{n=1}^Nf_{\alpha_n}
\sim_*
f_{\alpha_1}\vee \sum_{n=2}^Nf_{\alpha_n}
\sim_*
f_{\alpha_1}\vee f_{\alpha_2}\vee \sum_{n=3}^Nf_{\alpha_n}
\sim_*\ldots\sim_*
\bigvee_{n=1}^Nf_{\alpha_n}
\le 1.
\end{align}
But then,
$\rho(\tfrac1N\sum_{1\le n\le N}f_{\alpha_n})
\le
\tfrac1N$
while 
$\tfrac1N\sum_{1\le n\le N}\rho(f_{\alpha_n})
\ge
\delta.$
\end{proof}

Lemma \ref{lemma m=1} shows that the complete 
and competitive financial market imagined by Arrow 
\cite{arrow} may not be feasible in a complex world,
that is with a large enough $\A$.
We have thus an instance in which complexity acts as 
a restriction to perfect competition%
\footnote{
The only theoretical contribution to the study of the 
link between competitive equilibria and complexity 
of which I am aware is the paper of Gale and Sabourian
\cite{gale_sabourian} -- and the references therein.
The authors prove that in a game theoretic, matching 
model rational agents with aversion to complexity end 
up playing a subgame perfect equilibrium which is 
perfectly competitive despite the finite number of 
players. Complexity refers here to strategies, and 
one strategy is considered to be more complex than 
another whenever the two coincide save on a set 
of states on which the former is constant.
}.
The Lemma also provides a negative answer to the
question of the existence of viable price systems,
as defined in section \ref{sec viable}. Essentially
this occurs because in the situation considered
the elements of the positive cone $\K$ defined
in \eqref{K} are too many and too diverse from
one another.

\begin{example}
\label{ex PP}
Let us return to Example \ref{ex EU}. Let each 
agent $\alpha$ in the economy be endowed 
with a probability prior $P_\alpha$, forming 
the collection $\Pred\subset\Prob(\A)$ (thus 
countably additive). Define $\ge_\Pred$ as in 
\eqref{ge PP}. If $\Pred$ is undominated then, 
by \cite[Theorem 3]{JMAA_2016}, there exists 
an uncountable, pairwise disjoint family of 
sets in $\A$ each of which is non negligible 
for some agent $\alpha$ - and thus for the 
economy as a whole. Denote the indicator of
each of these sets by $f_\alpha$. The convex 
cone $\K$ contains then collection 
$\mathcal F
=
\{f_\alpha:\alpha\in\mathfrak A\}$. 
By Lemma \ref{lemma m=1} and Theorem 
\ref{th viable} we know that $\Meas(\pi,\A)=\emp$ 
and that, as a consequence, $\pi$ is not $*$-viable. 
Nevertheless, no agent $\alpha$ will consider the 
entire set $\K$ as relevant for his choices. In fact 
most elements of $\mathcal F$ are negligible from 
his point of view, i.e. relatively to $P_\alpha$. In 
other words, no agent respects the condition 
\eqref{adm monotone} (and a fortiori \eqref{adm cts})
relatively to $\K$. Nevertheless, for each 
$\alpha\in\mathfrak A$ there is a preference system 
of class $\mathbf A_{f_\alpha}$ and thus it is certainly 
true that an equilibrium supporting price $\pi$ must
be locally $*$-viable although it cannot be linear 
(i.e. competitive). This example complements the 
results obtained by Bouchard and Nutz \cite{bouchard_nutz}.
\end{example}

Those extensions that satisfy the {\it NFM} condition 
have further, special mathematical properties. 


\begin{theorem}
\label{th uac}
Let $\rho\in\Pi(L(\Xa,\A))$ satisfy \eqref{NFM}. Then 
there exists $\mu\in\Meas(\rho,\A)$ such that
\begin{equation}
\label{uac}
\lim_{\mu(\abs f)\to0}\rho(\abs f\wedge 1)
	=
0.
\end{equation}
\end{theorem}

\begin{proof}
Of course $\abs f\wedge1\in\B(\Omega,\A,\Neg_*)$ 
for each $f\in\Fun{\Omega,\A}$ so that \eqref{uac}
is well defined. For each $\alpha$ in a given 
set $\mathfrak A$, let $\seqn{A^\alpha}$ be a 
decreasing sequence of sets in $\A$ 
satisfying the following properties:
\tiref i
for each distinct pair $\alpha,\beta\in\mathfrak A$ 
there exists $n(\alpha,\beta)\in\N$ such that
\begin{equation}
\label{disjoint}
A^\alpha_n\cap A^\beta_n=\emp
\qquad
n>n(\alpha,\beta)
\end{equation}
and
\tiref{ii}
for each $\alpha\in\mathfrak A$ there exists
$m_\alpha\in\Meas_0(\rho,\A)$ such that 
$\lim_nm_\alpha(A^\alpha_n)>0$. 
If the set $\mathfrak A$ is uncountable, then, as 
in the preceding Lemma \ref{lemma m=1}, we can 
fix $\delta>0$ and extract a sequence 
$\alpha_1,\alpha_2,\ldots\in\mathfrak A$ such 
that, letting $f_n^i=\set{A_n^{\alpha_i}}$,
\begin{equation}
\label{pos}
\inf_{i\in\N}\lim_{n\to+\infty}\rho\big(f^i_n\big)
	>
\delta.
\end{equation}
For each $k\in\N$ define 
$n(k)=1+\sup_{\{i,j\le k:i\ne j\}}n(\alpha_i,\alpha_j)$.
Then $f^1_{n(k)},\ldots,f^k_{n(k)}\in\B(\Omega,\A)$ are 
pairwise disjoint functions with values in $[0,1]$ and
 such that
\begin{equation}
\inf_{1\le i\le k}\rho\big(f^i_{n(k)}\big)
	>
\delta.
\end{equation}
But then, taking $w_i=1/k$, we obtain
\begin{equation}
\sum_{i=1}^k\rho\big(w_if^i_{n(k)}\big)
	>
\delta
\qtext{while}
\rho\Big(\sum_{i=1}^kw_if^i_{n(k)}\Big)
	=
\frac1k\rho\Big(\sum_{i=1}^kf^i_{n(k)}\Big)
	\le
\frac1k
\end{equation}
so that $\ma(\rho)=1$, contradicting our initial 
assumption. We thus reach the conclusion that 
$\mathfrak A$ must be countable and deduce 
from this and from \cite[Theorem 2]{JMAA_2019} 
that $\Meas_0(\rho,\A)$ is dominated by some of its 
elements: let this be $\mu$. In addition, 
$\Meas_0(\rho,\A)$ is weak$^*$ compact as a subset 
of $ba(\A)$, as proved in Theorem \ref{th M}. 
It follows from \cite[Theorem 1.3]{zhang} that 
$\Meas_0(\rho,\A)$ is weakly compact. 

Take a sequence $\seqn E$ in $\A$ such that 
$\mu(E_n)\le2^{-n}$ and for each $n\in\N$
choose $m_n\in\Meas_0(\rho,\A)$ such that
\begin{equation*}
m_n(E_n)
	=
\sup_{m\in\Meas_0(\pi,\A)}m(E_n).
\end{equation*}
Passing to a subsequence if required, we can assume, 
by virtue of the Eberlein-\v{S}mulian Theorem
\cite[V.6.1]{bible}, that $\seqn m$ is weakly convergent 
and so, by the finitely additive version of the Theorem of 
Vitali, Hahn and Saks (see e.g. \cite[Theorem 8.7.4]{rao}), 
that the set $\{m_n:n\in\N\}$ is uniformly absolutely 
continuous with respect to $m_0=\sum_n2^{-n}m_n$. 
However, given that $m_0\in\Meas_0(\pi,\A)$, we also 
have $\mu\gg m_0$ so that $\{m_n:n\in\N\}$ is uniformly
absolutely continuous with respect to $\mu$ as well.
Thus,
\begin{equation}
\lim_{\mu(A)\to0}\sup_{m\in\Meas_0(\rho,\A)}m(A)
	=
0.
\end{equation}
Let $\seqn f$ be a sequence in $L(\Xa,\A)$ that converges
to $0$ in $L^1(\Omega,\A,\mu)$, and therefore in $\mu$ measure.
Then, by \eqref{attain} and for arbitrary $0<c<1$,
\begin{align*}
\lim_n\rho(\abs{f_n}\wedge 1)
	\le
\lim_n\rho\big(\sset{\abs{f_n}>c}\big)
+
c
	=
\lim_n
\sup_{m\in\Meas_0(\rho,\A)}m(\abs{f_n}>c)+c
\end{align*}
so that the claim follows. 
\end{proof}

To highlight the importance of this last claim, we note
that the absolute continuity property for charges, even 
if defined on a $\sigma$ algebra cannot be simply stated 
in terms of null sets. As a consequence, the existence 
of a strictly positive element of $\Meas_0(\pi,\A)$ 
established in Theorem \ref{th NFM}, is not sufficient to 
imply that the set $\Meas_0(\pi,\A)$ is dominated, i.e. that 
each of its elements is absolutely continuous with respect 
to a given one. It rather induces the weaker conclusion that 
there is a given pricing measure $m_0$ such that 
$m_0(A)=0$ implies $m(A)=0$ for all $m\in\Meas_0(\pi,\A)$.

On the other hand, if such a dominating element exists 
then, by weak compactness, it dominates $\Meas_0(\pi,\A)$ 
uniformly. A similar conclusion is not true in the countably 
additive case treated in the traditional approach. In that 
approach, the fact that risk neutral measures are dominated 
is an immediate consequence of the assumption of a 
given, reference probability measure but the set of such
measures is not weakly$^*$ compact when regarded as 
a subset of $ba(\A)$. This 
special feature illustrates a possible advantage of the 
finitely additive approach over the countably additive 
one. 

Eventually, notice that Theorem \ref{th uac} does not 
require the no arbitrage property and may thus be 
adapted to the case in which $L(\Xa,\A)$ is a generic 
vector lattice of functions on $\Omega$ containing 
the bounded functions and $\rho$  a monotonic,
subadditive and cash additive function, such as
the Choquet integral with respect to a sub modular
capacity.

Another characterization of the condition $\Meas(\pi,\A)\ne\emp$
may be obtained as follows:

\begin{theorem}
\label{th L1}
A market $(\Xa,\ge_*,\pi)$ satisfies the condition
$\Meas(\pi,\A)\ne\emp$ if and only if there exists 
$P\in\Prob_{ba}(\A,\Neg_*)$ such that
\begin{equation}
\label{NFLVR(P)}
\Xa\subset L^1(\Omega,\A,P)
\qtext{and}
\cl[L^1(\Omega,\A,P)]{\CapiA}
\cap
\{f\in\Fun{\Omega,\A}:f>_*0\}
=
\emp.
\end{equation}
If, conversely, $P\in\Prob_{ba}(\A,\Neg_*)$ satisfies 
\eqref{NFLVR(P)} then $\Neg_*$ coincides with the 
collection $\Neg_{_P}$ of $P$ null sets.
\end{theorem}

\begin{proof}
If $\mu\in\Meas(\pi,\A)$ then, by definition, $\Xa\subset L^1(\Omega,\A,\mu)$
and $\int fd\mu\le0$ for each 
$f
	\in
\cl[L^1(\Omega,\A,\mu)]{\CapiA}$
which rules out $f>_*0$. Choose $P=\mu$. Conversely, if 
$P\in\Prob_{ba}(\A,\Neg_*)$ satisfies \eqref{NFLVR(P)}
and $h>_*0$, then $h\wedge 1\in L^1(\Omega,\A,P)$ and
$h\wedge 1>_*0$. There exists then a positive and 
continuous linear functional $\phi_h$ on $L^1(\Omega,\A,P)$ 
such that
\begin{equation}
\label{phi_h}
\sup\Big\{\phi_h(f):f\in\cl[L^1(\Omega,\A,P)]{\CapiA}\Big\}
	\le
0
	<
\phi_h(h\wedge 1).
\end{equation}
Given that necessarily $\phi_h(1)>0$, \eqref{phi_h}
remains unchanged if we replace $\phi_h$ by its
normalization so that we can assume $\phi_h(1)=1$.
This implies that $\phi_h\in\Ext_0(\pi,\A)$ and, by 
Theorem \ref{th rep}, that $\phi_h$ admits the 
representation 
\begin{equation}
\phi_h(f)
	=
\int fdm_h
\quad
f\in L^1(\Omega,\A,\mu)
\end{equation}
for some $m_h\in\Meas_0(\pi,\A)$ that satisfies the
inclusion
$L^1(\Omega,\A,P)
	\subset 
L^1(\Omega,\A,m_h)$
and thus such that $m_h\ll P$. But then, if $h=\set A$ 
and $A\notin\Neg_*$, we conclude that $P(A)>0$. 
Moreover, by exploiting the finitely additive version of 
Halmos and Savage theorem, \cite[Theorem 1]{JMAA_2016}, 
we obtain
that the set $\{m_h:h\in\Fun{\Omega,\A}, h>_*0\}$ is dominated
by some $m_0\in\Meas_0(\pi,\A)$. It is then clear
that $m_0(f\wedge 1)>0$ for all $f>_*0$ and thus
that $m_0\in\Meas(\pi,\A)$.
\end{proof}

Implicitly, Theorem \ref{th L1} provides an answer to
the question raised in subsection \ref{subsec rationality}
relative to the conditions under which the ranking
$\ge_*$ takes the form $\ge_P$ for some
$P\in\Prob_{ba}(\A)$.


%

\section{Countably Additive Markets}
\label{sec c.a.}
Given the emphasis on countable additivity which 
dominates the traditional financial literature, it is 
natural to ask if it possible to characterise those 
markets in which the set $\Meas_0(\pi,\A)$ contains 
a countably additive element. A more ambitious 
question is whether such measure is strictly positive, 
i.e. an element on $\Meas(\pi,\A)$. 

Not surprisingly, an exact characterisation may be
obtained by considering the fairly unnatural possibility 
of forming portfolios which invest simultaneously in 
countably many different assets. This induces to 
modify the definition \eqref{market power} into the 
following (again with the convention $0/0=0$):
\begin{equation}
\na(\rho;f_1,f_2,\ldots)
	=
\lim_{k\to+\infty}
\frac
{\sum_{n\le k}\rho(f_n)-\rho\big(\sum_nf_n\big)}
{\sum_{n\le k}\rho(f_n)}
\qquad
\rho\in\Pi\big(L(\Xa,\A)\big)
\end{equation}
for all sequences $f_1,f_2,\ldots\in\B(\Omega,\A)_+$ 
such that $\sum_nf_n\in\B(\Omega,\A)$. 

It may at first appear obvious that, upon buying separately 
each component of a given portfolio, the investment cost 
results higher, but considered more carefully, this is indeed 
correct only if the infinite sum $\sum_n\rho(f_n)$ corresponds 
to an actual cost, i.e only if such a strategy of buying separately
infinitely many assets is feasible on the market.

Define then the functional
\begin{equation}
\na(\rho)
	=
\inf
\na(\rho;f_1,f_2,\ldots)
\end{equation}
where the infimum is computed with respect to all sequences 
in $\B(\Omega,\A)_+$ with bounded sum. Notice that in 
principle, the inequality 
$\na(\rho;f_1,f_2,\ldots)
	\ge
0$
is no longer valid while, of course, 
$\na(\rho)
	\le 
\ma(\rho)$. 

\begin{theorem}
\label{th ca}
Let $\pi\in\Pi_0(\Xa)$. Then: 
\begin{enumerate}[(a).]
\item
$\Meas_0(\pi,\A)\cap\Prob(\A)\ne\emp$
if and only if there exists $\rho\in\Ext_0(\pi,\A)$
such that 
$\na(\rho)>-\infty$ and that
\begin{equation}
\label{bdd}
\sum_n\rho(f_n)
	<
\infty
\qtext{for all}
f_1,f_2,\ldots\in\B(\Omega,\A)_+
\qtext{with}
\sum_nf_n\in\B(\Omega,\A);
\end{equation}
\item
$\Meas(\pi,\A)\cap\Prob(\A)\ne\emp$
if and only if there exists $\rho\in\Ext(\pi,\A)$
such that 
$1>\ma(\rho)\ge\na(\rho)>-\infty$.
\end{enumerate}
\end{theorem}

\begin{proof}
Of course any element 
$m
	\in
\Meas_0(\pi,\A)\cap\Prob(\F)$
when considered as a pricing function is an element of
$\Ext_0(\pi,\A)$ such that
$\na(m)
	=
\ma(m)
	=
0$.
This proves necessity for both claims. To prove 
sufficiency, let $\rho\in\Pi_0(L(\Xa,\A))$ and choose 
a sequence $\seqn f$ in $\B(\Omega,\A)_+$ with 
$\sum_nf_n\in\B(\Omega,\A)$.
If $\ma(\rho)<1$, then
\begin{equation}
\sum_n\rho(f_n)
	=
\lim_k\sum_{n\le k}\rho(f_n)
	\le
\frac1{1-\ma(\rho)}\lim_k\rho\Big(\sum_{n\le k}f_n\Big)
	\le
\frac1{1-\ma(\rho)}\rho\Big(\sum_nf_n\Big)
	<
\infty
\end{equation}
so that \eqref{bdd} is satisfied. It is therefore enough to
show that if $\rho$ meets the conditions listed under
\tiref a then $\Meas_0(\rho,\A)\subset\Prob(\A)$.
Choose $m\in\Meas_0(\rho,\A)$ let $\seqn A$ be a 
disjoint sequence in $\A$ and let 
$f_n=\set{A_n}$. We get $\sum_n\rho(f_n)<\infty$
and therefore
\begin{align*}
m\Big(\bigcup_nA_n\Big)
	&=
\sum_nm(A_n)+\lim_km\Big(\bigcup_{n>k}A_n\Big)
	\\&\le
\sum_nm(A_n)+\lim_k\rho\Big(\sum_{n>k}f_n\Big)
	\\&\le
\sum_nm(A_n)+[1-\na(\rho)]\lim_k\sum_{n>k}\rho(f_n)
	\\&=
\sum_nm(A_n).
\end{align*}
Thus, $\Meas_0(\rho,\A)\subset\Prob(\A)$. 
Claim \tiref{b} follows from the preceding remarks and 
Theorem \ref{th NFM}.
\end{proof}

The conditions for the existence of a countably additive
pricing measure listed under \tiref a and \tiref b are
perhaps deceptively simple. They are in fact increasingly
restrictive the finer the $\sigma$ algebra $\A$. With
$\A$ equal to the power set of $\Omega$, virtually 
every sequence of positive and bounded functions
may potentially produce a violation of the condition
$\na(\rho)>-\infty$. To stress this point, we observe 
that such inequality implies that, whenever $\seqn f$ 
is a uniformly bounded sequence of negligible functions, 
then necessarily $\sup_nf_n$ has to be negligible as 
well. An obvious implication is that $\Neg_*$ has to 
be closed with respect to countable unions, a property 
which requires a rather deep reformulation of the axioms 
\eqref{ge} that characterize economic rationality, as embodied into
$\ge_*$. There may well be cases in which such additional 
condition is simply contradictory and which cast doubts on
the economic adequacy of the countably additive paradigm. 
A model in which the sample space $\Omega$ is a separable 
metric space is a good case in point. In fact, if we take 
$\Neg_*$ to consist of sets of first category -- which
would clearly be a good example of what people
consider as negligible -- then, as is well known, 
$\Prob(\A,\Neg_*)
	=
\emp$,
see \cite[Th\'eor\`eme 1]{szpilrajn}.

\appendix

\section{Auxiliary results}

\subsection{The finitely additive integral}
In order for this work to be as self contained as
possible, we recall from \cite{bible} and \cite{rao} 
some basic facts about the finitely additive integral. 
Let then $m\in ba(\A)_+$.

\begin{definition}
$f\in\Fun\Omega$ is said to be integrable with 
respect to $m$ -- in symbols $f\in L^1(\Omega,\A,m)$ 
-- if and only if there exists a sequence 
$\seqn f$ of $\A$ measurable, simple functions such 
that each $f_n$ is integrable, $f_n$ converges to $f$
in $m$-measure and
\begin{equation}
\lim_{m,n}\int\abs{f_n-f_m}dm=0
\end{equation}
\end{definition}

\begin{theorem}
Let $\seqn h$ be a sequence in $L^1(\Omega,\A,m)$
and $h\in\Fun\Omega$. Then $h\in L^1(\Omega,\A,m)$
and $h_n$ converges to $h$ in the norm of 
$L^1(\Omega,\A,m)$ if and only if 
(i)
$h_n$ converges to $h$ in $m$-measure and
(ii)
$\lim_{m(A)\to0}\sup_n\int_A\abs{f_n}dm=0$.
\end{theorem}

\subsection{Integral representation}
Because of its repeated use in several proofs in this 
paper, we restate for ease of reference the following 
general representation Theorem which is just an
adaptation of \cite[Theorem 3.3]{JCA_2018} to the
special case considered in this work.

\begin{theorem}
\label{th rep}
Let $L$ be a vector sublattice of $\Fun{\Omega,\A}$ 
and $\phi$ a positive linear functional defined thereon.
There exists a positive linear functional $\phi^\perp$
on $L$ and a positive set function 
$\lambda\in ba(\Omega,\A)_+$
such that
(i)
$\phi^\perp(f)=0$ for all $f\in L\cap\B(\Omega,\A)$;
(ii)
$L\subset L^1(\Omega,\A,\lambda)$
and
\begin{equation}
\label{phi rep}
\phi(f)
	=
\phi^\perp(f)
+
\int fd\lambda
\qquad
f\in L.
\end{equation}
If $L$ is such that $f\wedge 1\in L$ for each $f\in L$, then
$\phi^\perp=0$ if and only if
\begin{equation}
\label{appr}
\phi(f)
	=
\lim_k\phi\big(f^+\wedge k-f^-\wedge k\big)
\qquad
f\in L.
\end{equation}
\end{theorem}

\begin{proof}
The proof of the existence of the representation 
\eqref{phi rep} for some $\lambda\in ba(\Omega)_+$ 
follows immediately from \cite[Theorem 3.3]{JCA_2018}: 
given that $\phi$ is assumed to be positive, the 
identity map on $L$ is clearly $\phi$ conglomerative 
and, since $L$ is a lattice, directed as well. Given that 
$L$ consists of $\A$ measurable functions, the representing 
charge $\lambda$ may be restricted to $\A$. Denote 
such restriction again by $\lambda$. If $f\wedge 1\in L$ 
for each $f\in L$, then 
$f_k
	=
f^+\wedge k-f^-\wedge k\in L$.
Since $f\in L^1(\Omega,\A,\lambda)$ we obtain from ordinary
results on finitely additive integrals, \cite[III.3.6]{bible},
that $\lim_k\int\abs{f-f_k}d\lambda=0$. Thus if 
$\phi^\perp=0$ then \eqref{appr} holds; conversely, 
\eqref{appr} implies
\begin{align*}
\phi^\perp(f)
	=
\lim_k\Big[\phi(f_k)-\int f_kd\lambda\Big]
	=
\lim_k\phi^\perp(f_k)
	=
0.
\end{align*}
\end{proof}

\BIB{acm}


\begin{thebibliography}{10}

\bibitem{acciaio_et_al}
{\sc Acciaio, B., Beiglb\"ock, M., Penkner, F., and Schachermayer, W.}
\newblock A model-free version of the fundamental theorem of asset pricing and
  the super-replication theorem.
\newblock {\em Math. Finan. 26}, 2 (2016), 233--251.

\bibitem{allen_gale}
{\sc Allen, F., and Gale, D.}
\newblock Arbitrage, short sales, and financial innovation.
\newblock {\em Econometrica 59}, 4 (1991), 1041--1068.

\bibitem{araujo_et_al}
{\sc Araujo, A., Chateauneuf, A., and Faro, J.~H.}
\newblock Pricing rules and {A}rrow-{D}ebreu ambiguous valuation.
\newblock {\em Econ. Theory 49\/} (2012), 1--35.

\bibitem{arrow}
{\sc Arrow, K.~J.}
\newblock The r\^ole of securities in the optimal allocation of risk bearing.
\newblock {\em Rev. Econ. Stud. 31\/} (1964), 91--96.

\bibitem{bewley}
{\sc Bewley, T.~F.}
\newblock Existence of equilibria in economies with infinitely many
  commodities.
\newblock {\em J. Econ. Theory 4}, 3 (1972), 514--540.

\bibitem{rao}
{\sc {Bhaskara Rao}, K. P.~S., and {Bhaskara Rao}, M.}
\newblock {\em Theory of Charges}.
\newblock Academic Press, London, 1983.

\bibitem{biais_glosten_spatt}
{\sc Biais, B., Glosten, L., and Spatt, C.}
\newblock Market microstructure: A survey of microfoundations, empirical
  results, and policy implications.
\newblock {\em J. Finan. Mar. 8}, 2 (2005), 217 -- 264.

\bibitem{bisin}
{\sc Bisin, A.}
\newblock General equilibrium with endogenously incomplete financial markets.
\newblock {\em J. Econ. Theory 82\/} (1998), 19--45.

\bibitem{bouchard_nutz}
{\sc Bouchard, B., and Nutz, M.}
\newblock Arbitrage and duality in nondominated discrete-time models.
\newblock {\em Ann. Appl. Probab. 25}, 2 (2015), 823--859.

\bibitem{burzoni_riedel_soner}
{\sc Burzoni, M., Riedel, F., and Soner, M.}
\newblock Viability and arbitrage under {K}nightian uncertainty.
\newblock {\em Econometrica\/} (forthcoming), 1--32.

\bibitem{MAFI_2008}
{\sc Cassese, G.}
\newblock Asset pricing with no exogenous probability measure.
\newblock {\em Math. Finance 18}, 1 (2008), 23--54.

\bibitem{JMAA_2016}
{\sc Cassese, G.}
\newblock The theorem of {H}almos and {S}avage under finite additivity.
\newblock {\em J. Math. Anal. Appl. 437\/} (2016), 870--881.

\bibitem{ECTH_2017}
{\sc Cassese, G.}
\newblock Asset pricing in an imperfect world.
\newblock {\em Econ. Th. 64\/} (2017), 539--570.

\bibitem{JCA_2018}
{\sc Cassese, G.}
\newblock Conglomerability and the representation of linear functionals.
\newblock {\em J. Convex Anal. 25\/} (2018), 789--815.

\bibitem{JMAA_2019}
{\sc Cassese, G.}
\newblock Control measures on {B}oolean algebras.
\newblock {\em J. Math. Anal. Appl. 478}, 2 (2019), 764--775.

\bibitem{conway}
{\sc Conway, J.~B.}
\newblock {\em A Course in Functional Analysis}.
\newblock Springer-Verlag, Berlin, 1990.

\bibitem{davis_hobson}
{\sc Davis, M. H.~A., and Hobson, D.~G.}
\newblock The range of traded option prices.
\newblock {\em Math. Finance 17}, 1 (2007), 1--14.

\bibitem{delbaen_schachermayer}
{\sc Delbaen, F., and Schachermayer, W.}
\newblock A general version of the fundamental theorem of asset pricing.
\newblock {\em Math. Ann. 300\/} (1994), 463--520.

\bibitem{dellacherie_meyer_A}
{\sc Dellacherie, C., and Meyer, P.-A.}
\newblock {\em Probabilities and Potential {A}}.
\newblock North-Holland, Amsterdam, 1978.

\bibitem{bible}
{\sc Dunford, N.~J., and Schwartz, J.~T.}
\newblock {\em Linear Operators. {P}art {I}}.
\newblock Wiley and Sons, New York, 1988.

\bibitem{el-karoui_ravanelli}
{\sc El-Karoui, N., and Ravanelli, C.}
\newblock Cash subadditive risk measures and interest rate ambiguity.
\newblock {\em Math. Finan. 19}, 4 (2009), 561--590.

\bibitem{epstein_ji}
{\sc Epstein, L.~G., and Ji, S.}
\newblock Ambiguous volatility and asset pricing in continuous time.
\newblock {\em Rev. Finan. Stud. 26}, 7 (2013), 1740--1786.

\bibitem{gale_sabourian}
{\sc Gale, D., and Sabourian, H.}
\newblock Complexity and competition.
\newblock {\em Econometrica 73}, 3 (2005), 739--769.

\bibitem{harrison_kreps}
{\sc Harrison, M.~J., and Kreps, D.~M.}
\newblock Martingales and arbitrage in multiperiod securities markets.
\newblock {\em J. Econ. Theory 20}, 3 (1979), 381--408.

\bibitem{jarrow_protter_shimbo}
{\sc Jarrow, R.~A., Protter, P., and Shimbo, K.}
\newblock Asset price bubbles in complete markets.
\newblock In {\em Advances in Mathematical Finance}, M.~C. Fu, R.~A. Jarrow,
  J.-Y.~J. Yen, and R.~J. Elliott, Eds. Birkh{\"a}user, Boston, MA, 2007,
  pp.~97--121.

\bibitem{jouini_kallal_99}
{\sc Jouini, E., and Kallal, H.}
\newblock Viability and equilibrium in securities markets with frictions.
\newblock {\em Math. Finance 9}, 3 (1999), 275--292.

\bibitem{kreps}
{\sc Kreps, D.~M.}
\newblock Arbitrage and equilibrium in economies with infinitely many
  commodities.
\newblock {\em J. Math. Econ. 8\/} (1981), 15--35.

\bibitem{loewenstein_willard_et}
{\sc Loewenstein, M., and Willard, G.~A.}
\newblock Local martingales, arbitrage, and viability. {F}ree snacks and cheap
  thrills.
\newblock {\em Econ. Theory 16}, 1 (2000), 135 -- 161.

\bibitem{maharam}
{\sc Maharam, D.}
\newblock An algebraic characterization of measure algebras.
\newblock {\em Ann. Math. 48\/} (1947), 154--167.

\bibitem{mas-colell}
{\sc Mas-Colell, A.}
\newblock The price equilibrium existence problem in topological vector
  lattices.
\newblock {\em Econometrica 54}, 5 (1986), 1039--1053.

\bibitem{mas-colell_zame}
{\sc Mas-Colell, A., and Zame, W.}
\newblock Equilibrium theory in infinite dimensional spaces.
\newblock In {\em Handbook of Mathematical Economics}, W.~Hildenbrand and
  H.~Sonnenschein, Eds., 1~ed., vol.~4. Elsevier, 1991, ch.~34, pp.~1835--1898.

\bibitem{radner}
{\sc Radner, R.}
\newblock Existence of equilibrium of plans, prices, and price expectations in
  a sequence of markets.
\newblock {\em Econometrica 40}, 2 (1972), 289--303.

\bibitem{riedel}
{\sc Riedel, F.}
\newblock Finance without probabilistic prior assumptions.
\newblock {\em Dec. Econ. Financ. 38\/} (2015), 75--91.

\bibitem{schmeidler_89}
{\sc Schmeidler, D.}
\newblock Subjective probability and expected utility without additivity.
\newblock {\em Econometrica 57\/} (1989), 571--587.

\bibitem{sion}
{\sc Sion, M.}
\newblock On general minimax theorems.
\newblock {\em Pacific J. Math. 8\/} (1958), 171--175.

\bibitem{szpilrajn}
{\sc Szpilrajn-Marczewski, E.}
\newblock Remarques sur les fonctions compl\`{e}tement additives d'ensemble et
  sur les ensembles jouissant de la propri\'{e}t\'{e} de {B}aire.
\newblock {\em Fund. Math. 22\/} (1934), 303--311.

\bibitem{zhang}
{\sc Zhang, X.-D.}
\newblock On weak compactness in spaces of measures.
\newblock {\em J. Func. Anal. 143\/} (1997), 1--9.

\end{thebibliography}
\end{document}